%% file: GrassBethearxiv.tex
\documentclass[paper, reqno]{amsart}

\numberwithin{equation}{section}
\usepackage{verbatim} 

\usepackage{amsfonts, amssymb, amsthm, amsmath, bbm, wasysym, MnSymbol,enumerate,pifont}
\usepackage[shortlabels]{enumitem}
\usepackage[all,cmtip]{xy}
\usepackage{youngtab}
\usepackage[dvipsnames]{xcolor}
\usepackage{mathtools}
\usepackage{multicol}

\newtheorem{theorem}{Theorem}[section]
\newtheorem{lemma}[theorem]{Lemma}
\newtheorem{proposition}[theorem]{Proposition}
\newtheorem{corollary}[theorem]{Corollary}
\theoremstyle{definition}
\newtheorem{definition}[theorem]{Definition}
\newtheorem{remark}[theorem]{Remark}

\newtheorem{example}[theorem]{Example}

\usepackage{tikz}
\usepackage{tikz-cd}
\usepackage[all]{xy}

\newcommand{\Xn}{{\op{X}}_n}
\newcommand{\Xd}{{\op{X}}_{\bf d}}
\newcommand{\CH}{\op{CoHa}(A_1)}
\newcommand{\cS}{\mathcal{S}}
\newcommand{\mV}{\mathbb{V}}
 
\newcommand{\bb}{\mathbf b} 
 \newcommand{\ba}{\mathbf a}
\newcommand{\bd}{\mathbf d} 
\newcommand{\oS}{\op{S}} 
\newcommand{\C}{\op{C}} 
\newcommand{\Func}{\op{Func}_n}
\newcommand{\eu}{\op{eu}}  
\newcommand{\inc}{\op{inc}}  
\newcommand{\incl}{\op{incl}}
\newcommand{\Supp}{\op{Supp}} 
\newcommand{\bw}{\op{w}}   
\newcommand{\Eins}{\mathbbm{1}}

\usepackage{pict2e}

\usepackage{pict2e}
\def\Cross1{%
  \unitlength=1 pt\thicklines
 \begin{picture}(30,30)
 \put(15,0){\color{red}\line(0,1){30}}

 \put(0,15){\line(1,0){30}}

  \end{picture}}

\newcommand{\op}{\operatorname}

\newcommand{\mC}{\mathbb{C}}
\newcommand{\ov}{\overline}

\newcommand{\la}{\lambda}
\newcommand{\bH}{{\bf H}}
\newcommand{\bP}{{\bf P}}
\newcommand{\End}{\op{End}}
\newcommand{\cU}{\mathcal{U}}
\newcommand{\cZ}{\mathcal{Z}}
\newcommand{\mg}{\mathfrak{gl}_2}

\author{Vassily Gorbounov, Christian Korff and Catharina Stroppel}
\address{V. G.: Institute of Mathematics, University of Aberdeen, UK and 
National Research University Higher School of Economics, Russian Federation. }
\email{vgorb10@gmail.com}

\address{C. K.: School of Mathematics and Statistics, Glasgow University, UK. } \email{christian.korff@glasgow.ac.uk}
\address{C. S.: Hausdorff Center of Mathematics, University of Bonn, Germany. }
\email{stroppel@math.uni-bonn.de}

\begin{document}
\title[Yang-Baxter algebras as convolution algebras]{Yang-Baxter algebras as convolution algebras:\\ The Grassmannian case}
\begin{abstract} 
We present a simple but explicit example of a recent development which connects quantum integrable models with Schubert calculus: there is a purely geometric construction of solutions to the Yang-Baxter equation and their associated Yang-Baxter algebras which play a central role in quantum integrable systems and exactly solvable lattice models in statistical physics. We consider the degenerate five-vertex limit of the asymmetric six-vertex model and identify its associated Yang-Baxter algebra as convolution algebra arising from the equivariant Schubert calculus of Grassmannians.  We show how our method can be used to construct (Schur algebra type) quotients of the current algebra  $\mathfrak{gl}_2[t]$ acting on the tensor product of copies of its evaluation representation $\mathbb{C}^2[t]$. Finally we connect it with the COHA for the  $A_1$-quiver.
\end{abstract}
\maketitle

\tableofcontents
\addtocontents{toc}{\protect\setcounter{tocdepth}{0}}
\vspace{-0.5cm}
\section*{Introduction}
Quantum groups first arose in the physics literature, particularly in the work
of the Faddeev school, from the so-called quantum inverse scattering
method, \cite{quantumscatt}, \cite{BIK}, \cite{Fadeev}, which had been developed to construct and solve quantum integrable systems. 
In their original form, quantum groups are associative algebras whose defining relations are expressed in terms of a matrix of constants (depending on the
integrable system under consideration) called a {\it quantum R-matrix}. It was
realised independently by Drinfeld  \cite{Drinfeld} and Jimbo \cite{Jimbo} around 1985, that
these algebras are Hopf algebras, which, in many cases, are deformations of universal enveloping algebras of Lie algebras.
Thus, although many of the fundamental papers on quantum groups are
written in the language of integrable systems, their properties are accessible
by more conventional mathematical techniques, such as the theory of topological and algebraic groups and of Lie algebras.

The task of placing quantum groups in the setup of representations of quivers and geometric representation theory, was initiated by several 
researchers including Ringel \cite{Ringel}, Beilinson, Lusztig, MacPherson \cite{BLM} and Ginzburg  \cite{Ginzburg} just to name a few. The quantum groups studied by Drinfeld and Jimbo were hereby interpreted as the cohomology of certain complex algebraic varieties or as algebras of functions on certain varieties defined over finite fields. The associative algebra structure comes in both cases from a {\it convolution product construction}, in other words, these quantum groups were realised as subalgebras of some convolution algebra. 

In the last years several groups of researchers suggested a further geometrisation of the theory around quantum R-matrices (for instance in the work of Nekrasov and Shatashvilli, \cite{NSh}, Braverman, Maulik and Okounkov, \cite{BMO}, \cite{MO}, \cite{O}, Schiffmann and Vasserot \cite{SV})  leading into a slightly different direction. All  these approaches provide a  geometric construction of (an action of) {\it $R$-matrices}. Arising naturally as solutions of the Yang-Baxter equations, $R$-matrices are the major ingredient of the construction of a {\it Yang-Baxter algebra} and therefore it allows a geometric construction of Yang-Baxter algebras, often realised inside geometric constructions of quantum groups.  

In the current paper we will present some new results in this direction, based on an explicit example. 
The starting point is the observation (Proposition~\ref{startingpoint} and Corollary~\ref{corstartingpoint}) that the transition matrix between two natural bases in the equivariant cohomology of a Grassmannian 
  is a solution to the Yang-Baxter equation.  These two bases are formed by the classes of the attracting manifolds of two different one-dimensional torus actions (obtained by choosing two regular cocharacters of the natural $n$-dimensional torus, separated by a wall of the Weyl chamber). This is parallel to the approach of Maulik and Okounkov, \cite{MO}, where the classical rational $R$-matrix serves as the transition matrix between different choices of {\it stable envelope bases}, a basis of the equivariant cohomology of the cotangent bundle to Grassmannians which naturally arises in the context of symplectic geometry and quiver varieties. 

We instead work with the Grassmannians themselves and construct an action of the  $R$-matrix (see Corollary~\ref{startingpoint}) on equivariant cohomologies. Working with the base space instead of the cotangent bundle might be viewed as passing to an interesting degeneration  of the existing constructions. Taking the $R$-matrix as an input for the so-called $\op{RTT}=\op{TTR}$-construction from \cite{FRT} we produce two {\it Yang-Baxter algebras}, (Definition~\ref{DefYBA} and Proposition~\ref{PropYBA}), which we will study in this paper. These Yang-Baxter algebras are both a degeneration of the so-called asymmetric {\it six vertex model} studied in statistical physics, see \eqref{6vertex}, which allows us to connect it with the combinatorics and with Bethe vectors known in the theory of integrable systems, see also Remark~\ref{degeneration}.

Our first main result (Theorem~\ref{main}) describes these Yang-Baxter algebras as an explicit subalgebra of a convolution algebra acting on the equivariant cohomology $H^*_T=\bigoplus_{k=0}^N H^*_T(\op{Gr}(k,N))$ of the {\it Grassmannian variety} of all linear subspaces in a fixed vector space of dimension $N$.  
The two algebras are very similar, one involves twists with the Chern classes of the tautological bundle, the other Chern classes of the quotient bundle, a symmetry which was not obvious to us from the algebraic definition.

The (ordinary and equivariant) cohomology algebra of Grassmannians has been studied under the name {\it Schubert calculus} for more than a hundred years, see e.g. \cite{Schubertcalc} for an overview. In particular, its (rational) cohomology comes with two distinct bases over $H^*(\op{pt})$ (the $T$-equivariant cohomology of a point), namely the so-called {\it Schubert basis} arising as the classes of the Schubert varieties (which are closures of cells in a natural cell decomposition or attracting cells for the canonical torus action), and on the other hand the {\it torus fixed point basis}. 

To make the connection with the theory of quantum integrable systems we identify the Schubert basis in $H^*_T$ with the {\it standard} or {\it spin basis} of the tensor product $V[t]^{\otimes N}$ for a two dimensional vector space $V=\mC^2$.  It turns out, see Corollary~\ref{Bethe},  that the  torus fixed point basis then corresponds to the {\it basis of Bethe vectors}, see \cite{Bax} and references therein, in the context of quantum integrable systems. Our results might provide some new insights into Schubert calculus. 

Working with the interplay of different bases of $V^{\otimes N}$ (as $\mathfrak{sl}_2$-module) related to the  geometry of Grassmannians is not new and was used for instance frequently as the guiding principle for categorified representations of the quantum group $U_q(\mathfrak{sl}_2)$, see e.g.  \cite{FK}, \cite{Savage}, \cite{FKS}. The focus there is however on the realisation of  {\it integral} bases (like the {\it canonical} bases of Lusztig and Kashiwara) and the combinatorics of Kazdhan-Lusztig polynomials. In these frameworks,  $V^{\otimes N}$ is either identified with the direct sum of certain vector spaces of functions related to the Grassmannians or with the Grothendieck group of the direct sum of the categories of perverse sheaves on the Grassmannians. In contrast, the basis of Bethe vectors  is not integral. 

Moreover, in our framework,  the {\it current algebra} $\mathfrak{gl}_2[t]$ appears (instead of $\mathfrak{sl}_2$) with its action on the tensor product $V[t]^{\otimes N}$ of its evaluation module for {\it generic evaluation parameter}, as for instance also in \cite{Varchenko} (in the classical situation of $\mathfrak{gl}_2$ the situation would become trivial). We connect the universal enveloping algebra of the current Lie algebra $\mathfrak{gl}_2[t]$ in Theorem~\ref{YBAandcurrent} with the two Yang-Baxter algebras. More precisely we realise the algebra of endomorphisms of $H^*_T$ (after some localisation) generated by these two algebras with a quotient of the universal enveloping algebra of  $\mathfrak{gl}_2[t]$ localised with respect to a certain central subalgebra. This localisation is directly connected with the denominators appearing in the Bethe vectors and in the localisation theorem for equivariant cohomology.  The connection is established via a Schur-Weyl duality type statement between this current Lie algebra and a subalgebra of the group algebra of the affine Weyl group, Corollary~\ref{UmodI}

This also illustrates that our setup can be viewed as an interesting degeneration of the situation described e.g. in \cite{MO} or \cite{GRV}, where the Yangian and the affine Hecke algebra instead of the loop algebra and the group algebra of the affine Weyl group occur. The case of Yangians (and more general quantum groups) with a similar approach to ours was developed parallel and independently to our work in \cite{SV}. Our contribution should be seen as supplementary, giving an {\it explicit} example using classical geometry and elementary tools. We focus on the specific Grassmannian varieties, since we like to have explicit and elementary formulae. Our principal approach can however be generalised in the obvious way to arbitrary partial flag varieties of type $A$.  

In the last section we finally connect our construction with cohomological Hall algebras (COHAs) in the sense of Kontsevich and Soibelman, \cite{KS}. We equip $H_T^*$ with an action of a (in fact the easiest possible) COHA and show in Theorem~\ref{ThmgammasYB} that its image is contained in the endomorphism algebra generated by the two Yang-Baxter algebras.
 
\subsection*{Acknowledgement}
We are grateful to C. De Concini, O. Foda, H. Franzen, L. Michalcea,  R. Rim{\'a}nyi, N. Reshethikin, Y. Soibelman, V. Tarasov, A. Varchenko and P. Zinn-Justin for sharing ideas and  knowledge generously, and T. Przezdziecki for comments on a draft version. We thank MPI and the HCM in Bonn where most of this research was done. The first author has been partially funded by the  Russian Academic Excellence Project 5-100. 
\\\\
Conventions: For the whole paper we fix the ground field $\mC$ of complex numbers and abbreviate $\otimes=\otimes_\mC$, $\op{Hom}=\op{Hom}_\mC$, $\End=\End_\mC$ etc. For $A$, $B\in\End(V)$ we denote by $BA$ the composition $BA(v)=B(A(v))$ for $v\in V$.  For a natural number $N$ set $[N]=\{1,2,\ldots,N\}$.

\addtocontents{toc}{\protect\setcounter{tocdepth}{1}}
\section{Diagrammatic from integrable systems}
\label{diagrams}
While the construction of the Yang-Baxter algebra from certain L-matrices is standard in the physics literature, we feel it is much less known to mathematicians. We therefore develop it here in detail adapted to our situation,  and also refer to \cite{Haz} and \cite{Bump} for a similar setup.
\begin{center}
{\it For the whole paper we fix a natural number $N\geq 2$.} 
\end{center}
\subsection{Lax matrices and monodromy matrices}
In the following let $V=\mC^2$ with a fixed basis $v_0, v_1$. For any $\mC$-algebra $S$ we denote by  $\op{Mat}_2(S)$ the vector space of $2\times 2$-matrices with entries in $S$.  Any such matrix defines an endomorphism of $\mC^2\otimes S$ in the $S$-basis $v_0,v_1$. We often view elements in  $\op{Mat}_2(\op{Mat}_2(S))$ as endomorphisms of  $V\otimes V\otimes S=\mC^2\otimes \mC^2\otimes S$ in the ordered basis $v_0\otimes v_0$, $v_0\otimes v_1$, $v_1\otimes v_0$, $v_1\otimes v_1$.  For a finite set of variables $X$ denote by $\mathbb{C}[X]$ the polynomial ring in the variables $x\in X$ and let $W[X]:=W\otimes \mathbb{C}[X]$ for any vector space $W$. Set $\bP=\mathbb{C}[t_1,t_2,\ldots, t_N]$ and $V_N=V^{\otimes N}$  
 and let $$\mV_N=V^{\otimes N}\otimes\bP=V^{\otimes N}[t_1,t_2,\ldots, t_N].$$ 
 Let $S_N$  be the symmetric group of order $N!$ generated by the simple transpositions $s_i=(i,i+1)$, $i\in [N-1]$ and $\mC[S_N]$ its group algebra. It acts on $\bP$ by algebra automorphisms $w( t_i)=t_{w(i)}$. We denote by $\bP^{S_N}$ the invariants, that is the subalgebra of $\bP$ consisting of symmetric polynomials.
 
For $0\leq n\leq N$ let $\Lambda_n=\{(\la_1,\ldots,\la_N)\in\{0,1\}^N\mid \sum_{i=1}^N{\la_i}=n\}$ be the set of words of length $N$ in the letters $0$, $1$  with $n$ $1$'s. Set $\Lambda=\bigcup_{n=0}^N\Lambda_n=\{0,1\}^N.$  In the following we identify the vector spaces
 $$(V[t])^{\otimes N}=V^{\otimes N}[t_1,t_2,\ldots, t_N]=\mV_N.$$
 by sending a vector $(w_1\otimes p_1(t))\otimes (w_2\otimes p_2(t))\otimes \cdots\otimes(w_N\otimes p_N(t))$ with $w_i\in V$ and $p_i(t)\in \mC[t]$ to $(w_1\otimes w_2\otimes w_N)\otimes (p_1(t_1)p_2(t_2)\cdots p_N(t_N))$.

\begin{definition}
We pick a matrix $L(x,t)\in\op{Mat}_2(\op{Mat}_2(\mathbb{C}[x,t]))$, called {\it Lax matrix}, that is 
\begin{eqnarray}
\label{Lmatrix}
L(x,t)&=&\begin{pmatrix}
A(x,t)&B(x,t)\\
C(x,t)&D(x,t)
\end{pmatrix}
\end{eqnarray}
with entries in $\op{Mat}_2(\mathbb{C}[x,t])$. It defines with our choice of basis a $\mC[x,t]$-linear endomorphism of $V\otimes V[x,t]$.
\end{definition}

Consider now the tensor product $V[x]\otimes V[t]^{\otimes N}$. We call $V[x]$ the $0$th factor and then number the remaining tensor factors from $1$ to $N$.

 \begin{definition}
The {\it Monodromy matrix} is the endomorphism 
\begin{eqnarray}
\label{mono0}
M(x,t_1,\ldots ,t_N)&=&L_{0N}(x,t)\cdots L_{02}(x,t) L_{01}(x,t),
\end{eqnarray}
of $V[x]\otimes V[t]^{\otimes N}$, where $L_{ij}$ is the endomorphism acting as $L(x,t)$ on the  tensor factors $i$ and $j$ and as the identity on the remaining factors. 
\end{definition}

In the {\it standard} $\bP[x]$-basis, or {\it spin basis}, defined as
\begin{eqnarray}
\label{basisvec}
v_{\la}=v_{\la_{1}}\otimes v_{\la_{2}}\otimes\cdots\otimes v_{\la_{N}}, &\text{where}& \la=(\la_1,\la_2,\ldots,\la_N)\in\Lambda,
\end{eqnarray}
$M(x,t_1,t_2,\ldots, t_N)$ can be viewed as a block $2\times 2$-matrix with blocks, $A=A(x,t_1,t_2,\ldots, t_N)$ etc. of size $2^N\times 2^N$,  i.e.
\begin{eqnarray}
\label{Mono}
M(x,t_1,\ldots, t_N)&=&\begin{pmatrix}
A(x,t_1,\ldots, t_N)&B(x,t_1,\ldots, t_N)\\
C(x,t_1,\ldots, t_N)&D(x,t_1,\ldots, t_N)
\end{pmatrix}.\quad
\end{eqnarray}
and the blocks are matrices of $\bP[x]$-linear operators acting on $V[x]\otimes\mV_N$.

To calculate the matrix entries of \eqref{Mono} and algebraic relations between them we use some diagrammatic calculus which is a common tool in the theory of integrable systems. 
For this we first identify the standard basis vectors  \eqref{basisvec} with their $\{0,1\}$-words $(\la_1, \la_2,\dots,\la_N)$ of length $N$.  A {\it crossing} is a diagram of the form as shown on the left of \eqref{crosses} built out of an ordered pair of directed line intervals intersecting in their midpoints.
The $4$ line segments obtained by removing the intersection point are called {\it edges}. The directions assigned to the lines split the edges into {\it inputs} and {\it outputs}. The order of the lines puts an order on the inputs and on the outputs. In the picture, the edge on the left is the {\it first input}, the one on the top the {\it second input}, the edge on the right is the {\it first output} and the one on the bottom the {\it second output}. We will often omit drawing the arrows and the ordering.

A {\it labelled crossing} is a crossing, where all the four edges are labelled by an element from $\{0,1\}$ which we also display by drawing the edges {\it solid} if they are labelled with $1$ and {\it dotted} if they are labelled with $0$. For example, the following two pictures represent the same labelled crossing with input pair  (solid, dotted) and output pair (dotted, solid).
\begin{equation}\label{crosses}
\begin{tikzpicture}
\draw[very thick] (-3.0,0) -- (-3.5,0) node[anchor=east]{\ding{172}}; 
\draw[very thick] (-3,0) -- (-3,0.5) node[anchor=south]{\ding{173}};
\draw[very thick,->] (-3,0) -- (-2.4,0);
\draw[very thick,->] (-3,0) -- (-3,-0.6);
\draw[ultra thick] (-0.5,0) -- (0,0); 
\draw[ultra thick, dotted] (0,0.5) -- (0,0);
\draw[ultra thick, dotted] (0,0) -- (0.5,0);
\draw[ultra thick] (0,0) -- (0,-0.5);
\draw[very thick] (3,0) -- (2.5,0) node[anchor=east]{1}; 
\draw[very thick] (3,0) -- (3,0.5) node[anchor=south]{0};
\draw[very thick] (3,0) -- (3.5,0) node[anchor=west]{0};
\draw[very thick] (3,0) -- (3,-0.5) node[anchor=north]{1};
\end{tikzpicture}
\end{equation}

The 16 possible labelled crossings are displayed in the matrix below. 
\\
\begin{eqnarray}
\label{sixteen}
\begin{pmatrix}
\begin{tikzpicture}
\draw[ultra thick,dotted] (-0.4,0) -- (0.4,0); 
\draw[ultra thick, dotted] (0,0.4) -- (0,-0.4);
\end{tikzpicture}
\;\;
\begin{tikzpicture}
\draw[ultra thick,dotted] (-0.4,0) -- (0,0); 
\draw[ultra thick] (0,0.4) -- (0,0);
\draw[ultra thick, dotted] (0,0) -- (0.4,0);
\draw[ultra thick,dotted] (0,0) -- (0,-0.4);
\end{tikzpicture}
\;\;
\begin{tikzpicture}
\draw[ultra thick] (-0.4,0) -- (0,0); 
\draw[ultra thick,dotted] (0,0.4) -- (0,0);
\draw[ultra thick, dotted] (0,0) -- (0.4,0);
\draw[ultra thick,dotted] (0,0) -- (0,-0.4);
\end{tikzpicture}
\;\;
\begin{tikzpicture}
\draw[ultra thick] (-0.4,0) -- (0,0); 
\draw[ultra thick] (0,0.4) -- (0,0);
\draw[ultra thick, dotted] (0,0) -- (0.4,0);
\draw[ultra thick,dotted] (0,0) -- (0,-0.4);
\end{tikzpicture}\\
\begin{tikzpicture}
\draw[ultra thick,dotted] (-0.4,0) -- (0,0); 
\draw[ultra thick,dotted] (0,0.4) -- (0,0);
\draw[ultra thick, dotted] (0,0) -- (0.4,0);
\draw[ultra thick] (0,0) -- (0,-0.4);
\end{tikzpicture}
\;\;
\begin{tikzpicture}
\draw[ultra thick,dotted] (-0.4,0) -- (0.4,0); 
\draw[ultra thick] (0,0.4) -- (0,-0.4);
\end{tikzpicture}
\;\;
\begin{tikzpicture}
\draw[ultra thick] (-0.4,0) -- (0,0); 
\draw[ultra thick,dotted] (0,0.4) -- (0,0);
\draw[ultra thick, dotted] (0,0) -- (0.4,0);
\draw[ultra thick] (0,0) -- (0,-0.4);
\end{tikzpicture}
\;\;
\begin{tikzpicture}
\draw[ultra thick] (-0.4,0) -- (0,0); 
\draw[ultra thick] (0,0.4) -- (0,0);
\draw[ultra thick, dotted] (0,0) -- (0.4,0);
\draw[ultra thick] (0,0) -- (0,-0.4);
\end{tikzpicture}\\
\begin{tikzpicture}
\draw[ultra thick,dotted] (-0.4,0) -- (0,0); 
\draw[ultra thick,dotted] (0,0.4) -- (0,0);
\draw[ultra thick] (0,0) -- (0.4,0);
\draw[ultra thick,dotted] (0,0) -- (0,-0.4);
\end{tikzpicture}
\;\;
\begin{tikzpicture}
\draw[ultra thick,dotted] (-0.4,0) -- (0.4,0); 
\draw[ultra thick] (0,0.4) -- (0,0);
\draw[ultra thick] (0,0) -- (0.4,0);
\draw[ultra thick,dotted] (0,0) -- (0,-0.4);
\end{tikzpicture}
\;\;
\begin{tikzpicture}
\draw[ultra thick] (-0.4,0) -- (0.4,0); 
\draw[ultra thick,dotted] (0,0.4) -- (0,0);
\draw[ultra thick,dotted] (0,0) -- (0,-0.4);
\end{tikzpicture}
\;\;
\begin{tikzpicture}
\draw[ultra thick] (-0.4,0) -- (0,0); 
\draw[ultra thick] (0,0.4) -- (0,0);
\draw[ultra thick] (0,0) -- (0.4,0);
\draw[ultra thick,dotted] (0,0) -- (0,-0.4);
\end{tikzpicture}\\
\begin{tikzpicture}
\draw[ultra thick,dotted] (-0.4,0) -- (0,0); 
\draw[ultra thick,dotted] (0,0.4) -- (0,0);
\draw[ultra thick] (0,0) -- (0.4,0);
\draw[ultra thick] (0,0) -- (0,-0.4);
\end{tikzpicture}
\;\;
\begin{tikzpicture}
\draw[ultra thick,dotted] (-0.4,0) -- (0,0); 
\draw[ultra thick] (0,0.4) -- (0,0);
\draw[ultra thick] (0,0) -- (0.4,0);
\draw[ultra thick] (0,0) -- (0,-0.4);
\end{tikzpicture}
\;\;
\begin{tikzpicture}
\draw[ultra thick] (-0.4,0) -- (0,0); 
\draw[ultra thick,dotted] (0,0.4) -- (0,0);
\draw[ultra thick] (0,0) -- (0.4,0);
\draw[ultra thick] (0,0) -- (0,-0.4);
\end{tikzpicture}
\;\;
\begin{tikzpicture}
\draw[ultra thick] (-0.4,0) -- (0.4,0); 
\draw[ultra thick] (0,0.4) -- (0,-0.4);
\end{tikzpicture}
\end{pmatrix}
\end{eqnarray}

A {\it (labelled) lattice diagram} is a diagram obtained by composing finitely many labelled crossings vertically and horizontally by gluing an output edge with an input edge of the same colour,  see e.g. \eqref{diag}, together with a total ordering on the lines. If not otherwise specified, the ordering is always taking the left inputs from bottom to top followed by the top inputs from left to right, see \eqref{diag}.

\begin{definition}
 A Lax matrix $L\in\op{Mat}_2(\op{Mat}_2(\mathbb{C}[x,t]))$ assigns to a labelled crossing a {\it (Boltzmann) weight}, namely the $(i,j)$-entry of the matrix $L$ if the crossing is as the $(i,j)$-entry in \eqref{sixteen}.
The {\it weight of a labelled lattice diagram} is then the product of all weights of the crossings, but evaluated at $x=x_a,t=t_b$ if the $a$-th row and $b$-th column cross. If there is only one row (resp. one column) then we usually keep the variable $x$ (resp. t) instead of evaluating at $x_1$ (resp. $t_1$).
 \end{definition}

\begin{example}
\label{exs}
In case of $L=L(x,t)$ respectively $L=L'(x,t)$ from \eqref{L1L2} below, we have the following: the weight of the purely black  crossing equals $1$ respectively $x-t$. The weight of the lattice diagram in the middle  of \eqref{threediag} equals  $x+t_2$ respectively $x-t_1$, whereas for the one on the right it is $1$ and $(x-t_1)(x-t_2)$ respectively. 
\end{example}

  \subsection{Calculating matrix entries}  
It is not hard to check now that the matrix entries in  \eqref{Mono} can be calculated as follows: consider a {\it $1$-row lattice of length $N$}, that is a labelled lattice diagram obtained by putting $N$ labelled crossings next to each other, e.g. for $N=3$,

\begin{eqnarray}
\label{diag}
\begin{tikzpicture}
\draw[ultra thick, dotted] (0,0) -- (-1,0) node[anchor=east]{\ding{172}}; 
\draw[ultra thick,dotted] (-0.5,0) -- (-0.5,0.5) node[anchor=south]{\ding{173}};
\draw[ultra thick] (-0.5,0) -- (-0.5,-0.5);
\draw[ultra thick] (0,0) -- (0,0.5) node[anchor=south]{\ding{174}};
\draw[ultra thick] (0,0) -- (1,0);
\draw[ultra thick] (0,0) -- (0,-0.5);
\draw[ultra thick] (0.5,0) -- (0.5,0.5) node[anchor=south]{\ding{175}};
\draw[ultra thick, dotted] (0.5,0) -- (0.5,-0.5);
\draw[ultra thick, dotted] (6,0) -- (5,0); 
\draw[ultra thick,dotted] (5.5,0) -- (5.5,0.5) node[anchor=south]{0};
\draw[ultra thick] (5.5,0) -- (5.5,-0.5) node[anchor=north]{1};
\draw[ultra thick] (6,0) -- (6,0.5) node[anchor=south]{1};
\draw[ultra thick] (6,0) -- (7,0);
\draw[ultra thick] (6,0) -- (6,-0.5) node[anchor=north]{1};
\draw[ultra thick] (6.5,0) -- (6.5,0.5) node[anchor=south]{1};
\draw[ultra thick, dotted] (6.5,0) -- (6.5,-0.5) node[anchor=north]{0};
\end{tikzpicture}
%
\end{eqnarray}

The top sequence of labels defines a $\{0,1\}$-word  $\la_{in}$, and thus a vector $v_{\la_{in}}\in V_N$, called the {\it input vector}; likewise, the bottom sequence of labels defines the {\it output vector} $v_{\la_{out}} \in V_N$. 
The two (left and right) outer horizontal edges specify one of the operators $A$, $B$, $C$, $D$ from  \eqref{Mono} via the assignments
$$A\leftrightarrow (0,0),\, B\leftrightarrow (1,0),\,C\leftrightarrow (0,1),\,D\leftrightarrow (1,1).$$
The corresponding Boltzmann weight is then a polynomial in $\bP[x]$, and for the Lax matrices \eqref{L1L2} of total degree at most $N$. From the definitions we deduce: 

\begin{proposition}
Let $O\in\{A,B,C,D\}$ and $v_{\la_{in}}, v_{\la_{out}}\in V^{\otimes N}$ standard basis vectors. 
Then the matrix coefficient $(O(v_{\la_{in}}),v_{\la_{out}})$ is equal to the sum of the weights of all  $1$-row lattice diagrams  of length $N$ whose external edge labelling is fixed and given by the triple ($O$, $v_{\la_{in}}$, $v_{\la_{out}}$).
\end{proposition}
\begin{example} Let us calculate the coefficient $c$ of $Cv_{(0,1,1)}$  at $v_{(0,1,0)}$. The labelling of the external edges is as displayed in (i) which for $L(x,t)$ or $L'(x,t)$ as in \eqref{L1L2} can be extended to a unique lattice diagram with non-zero weight as in (ii), and thus  $x+t_2$ respectively $x-t_1$ is its weight.

\begin{equation}
\label{threediag}
\begin{tikzpicture}
\draw[ultra thick, dotted] (-0.5,0) -- (-1,0) node[anchor=east]{(i)\qquad}; 
\draw[ultra thick,dotted] (-0.5,-0.5) -- (-0.5,0.5);
\draw[ultra thick] (0,-0.5) -- (0,0.5);
\draw[ultra thick] (0.5,0) -- (0.5,0.5);
\draw[ultra thick] (0.5,0) -- (1,0);
\draw[ultra thick, dotted] (0.5,0) -- (0.5,-0.5);
\end{tikzpicture}
\qquad\qquad
\begin{tikzpicture}
\draw[ultra thick, dotted] (0.5,0) -- (-1,0) node[anchor=east]{(ii)\qquad}; 
\draw[ultra thick,dotted] (-0.5,-0.5) -- (-0.5,0.5);
\draw[ultra thick] (0,-0.5) -- (0,0.5);
\draw[ultra thick] (0.5,0) -- (0.5,0.5);
\draw[ultra thick] (0.5,0) -- (1,0);
\draw[ultra thick, dotted] (0.5,0) -- (0.5,-0.5);
\end{tikzpicture}
\qquad\qquad
\begin{tikzpicture}
\draw[ultra thick, dotted] (0.5,0) -- (-1,0) node[anchor=east]{(iii)\qquad}; 
\draw[ultra thick,dotted] (-0.5,-0.5) -- (-0.5,0.5);
\draw[ultra thick,dotted] (0,-0.5) -- (0,0.5);
\draw[ultra thick] (0.5,0) -- (0.5,0.5);
\draw[ultra thick] (0.5,0) -- (1,0);
\draw[ultra thick, dotted] (0.5,0) -- (0.5,-0.5);
\end{tikzpicture}
\end{equation}
 In comparison, the coefficient of $Cv_{(0,0,1)}$  at $v_{(0,0,0)}$ equals the weight of  (iii), see Example~\ref{exs}.
\end{example}

\subsection{Yang-Baxter equation}

\begin{definition} A pair $(R(x,y), L(x,t))$ from $\op{Mat}_2(\op{Mat}_2(\mathbb{C}[x,t]))$ is {\it a solution of the Yang-Baxter equation} if the following holds
\begin{eqnarray}
\label{RLL}
L_{23}(x_2,t)L_{13}(x_1,t)R_{12}(x_1,x_2)&=&R_{12}(x_1,x_2)L_{13}(x_1,t)L_{23}(x_2,t)\quad
\end{eqnarray}
as endomorphisms of $V^{\otimes 3}[x_1,x_2,t]$. The subindices indicate the factors on which the respective operator acts non-trivially. 
\end{definition}
By incorporating $R$, the diagrammatics from above can easily be adopted to express \eqref{RLL} graphically. To write the identity \eqref{RLL} as a system of  equalities for the matrix entries of the above $4\times 4$ matrices with blocks of size $2\times 2$, attach to $R(x,y)$ a St. Andrew's cross  {\huge${\times}$}
with first input at the bottom left and the second input at the top left, and consider the following two diagrams.

\[
\begin{tikzpicture}[scale=0.8]
\draw (-4.5,0) node{(LHS)};
\draw[very thick, rounded corners=8pt] (0,0.5) -- (-2,0.5) -- (-3,-0.5) node[anchor=east]{\ding{172}};
\draw[very thick, rounded corners=8pt] (0,-0.5) -- (-2,-0.5) -- (-3,0.5) node[anchor=east]{\ding{173}};
\draw[very thick] (-1,-1) -- (-1,1) node[anchor=south]{{\ding{174}}};
\end{tikzpicture}
\qquad\qquad
\begin{tikzpicture}[scale=0.8]
\draw (-4.5,0) node{(RHS)};
\draw[very thick, rounded corners=8pt] (0,0.5) -- (-1,-0.5) -- (-3,-0.5) node[anchor=east]{\ding{172}};
\draw[very thick, rounded corners=8pt] (0,-0.5) -- (-1,0.5) -- (-3,0.5) node[anchor=east]{\ding{173}};
\draw[very thick] (-2,-1) -- (-2,1) node[anchor=south]{{\ding{174}}};
\end{tikzpicture}
\]

Here the three lines correspond to the three tensor factors involved in \eqref{RLL},                                                                                                                                            with the inputs numbered as indicated. Reading the diagrams from left to right and top to bottom, the crossings correspond precisely to the factors involved in the compositions \eqref{RLL} (with the inputs always on the left and at the top and the outputs on the bottom and to the right). By a labelling of such a diagram we mean, as before, a diagram with labels $\{0,1\}$ attached to all the edges, and the weight of such a diagram is again defined as the product of the weights attached to the vertices (except that now the new type of crossings get weights from the matrix $R$). \\
The following is a consequence of the definitions:

\begin{proposition} 
\label{RL}
A pair $(R(x,y),L(x,t))$ from $\op{Mat}_2(\op{Mat}_2(\mathbb{C}[x,t]))$ gives  a solution of the Yang-Baxter equation \eqref{RLL} if and only if the following holds: for each common labelling of the external edges of the diagrams (LHS) and (RHS) the respective sums of weights (for all possible complete labellings) agree.
\end{proposition}

\begin{proposition} 
\label{RMM}
If $(R(x,y),L(x,t))$ is a solution to the Yang-Baxter equation, then the  identity
\begin{eqnarray*}
&&M_{2}(x_2,t_1,\ldots,t_N)M_{1}(x_1,t_1,\ldots,t_N)R_{12}(x_1,x_2)\\
&=&
R_{12}(x_1,x_2)M_{1}(x_1,t_1,\ldots,t_N)M_{2}(x_2,t_1,\ldots,t_N)
\end{eqnarray*}
holds in $\op{End}_{P[x_1,x_2]}(V\otimes V\otimes \mV_N[x_1,x_2])$ 
for any $N\geq 1$, where $M_i$ acts on the $i$th factor $V$  and on $\mV_N[x_1,x_2]$ as in \eqref{mono0}.
\end{proposition}

\begin{proof} This follows from Proposition~\ref{RL} by the usual rail road principle, which means diagrammatically that we move the crossing representing $R$ to the right
\[
\begin{tikzpicture}[scale=0.7]
\draw (0.5,0) node{=};
\draw[very thick, rounded corners=8pt] (0,0.5) -- (-2,0.5) -- (-3,-0.5); 
\draw[very thick, rounded corners=8pt] (0,-0.5) -- (-2,-0.5) -- (-3,0.5); 
\draw[very thick] (-0.3,-1) -- (-0.3,1); 
\draw[very thick] (-0.9,-1) -- (-0.9,1);
\draw[very thick] (-1.5,-1) -- (-1.5,1);
\end{tikzpicture}
\quad
\begin{tikzpicture}[scale=0.7]
\draw[very thick, rounded corners=8pt] (0,0.5) -- (-1.5,0.5) -- (-2,-0.5) -- (-3,-0.5); 
\draw[very thick, rounded corners=8pt] (0,-0.5) -- (-1.5,-0.5) -- (-2,0.5) -- (-3,0.5); 
\draw[very thick] (-0.5,-1) -- (-0.5,1); 
\draw[very thick] (-1,-1) -- (-1,1);
\draw[very thick] (-2.5,-1) -- (-2.5,1);
\end{tikzpicture}
\quad
\begin{tikzpicture}[scale=0.7]
\draw (-4,0) node{=\quad$\cdots$\quad=};
\draw[very thick, rounded corners=8pt] (0,0.5) -- (-1,-0.5) -- (-3,-0.5); 
\draw[very thick, rounded corners=8pt] (0,-0.5) -- (-1,0.5) -- (-3,0.5); 
\draw[very thick] (-2.7,-1) -- (-2.7,1); 
\draw[very thick] (-2.1,-1) -- (-2.1,1);
\draw[very thick] (-1.5,-1) -- (-1.5,1);
\end{tikzpicture}
\]
by applying iteratively the Yang-Baxter equation from Proposition~\ref{RL}.
\end{proof}
\subsection{Six and five vertex models}

Usually one restricts to $R$-matrices which preserve the weight spaces if we view $V$ as the vector representation of $\mathfrak{sl}_2(\mC)$. Matrices of this form
\begin{eqnarray}
\label{6vertex}
\begin{pmatrix}
\omega_1&0&0&0\\
0&\omega_3&\omega_5&0\\
0&\omega_6&\omega_4&0\\
0&0&0&\omega_2
\end{pmatrix}
\end{eqnarray}
appear for instance in the  {\it  asymmetric six vertex model}, \cite{Bax}, in the physics literature. 

\begin{definition}
\label{ex}
The two examples of $R$-matrices and Lax matrices, which will be treated here explicitly, are the following
\small
\begin{eqnarray*}
R(x,y):=
\begin{pmatrix}
1&0&0&0\\
0&x-y&1&0\\
0&1&0&0\\
0&0&0&1
\end{pmatrix}
&&
R'(x,y)=
\begin{pmatrix}
1&0&0&0\\
0&0&1&0\\
0&1&y-x&0\\
0&0&0&1
\end{pmatrix}\;
\end{eqnarray*}
\begin{eqnarray}
\label{L1L2}
L(x,t)=
\begin{pmatrix}
1&0&0&0\\
0&x+t&1&0\\
0&1&0&0\\
0&0&0&1
\end{pmatrix}
&&
L'(x,t)=
\begin{pmatrix}
x-t&0&0&0\\
0&1&1&0\\
0&1& 1&0\\
0&0&0&0
\end{pmatrix}
\end{eqnarray}

\end{definition}
The Lax operators \eqref{L1L2}  are so-called {\it 5-vertex degenerations}, see also Remark~\ref{degeneration}, of
the asymmetric 6-vertex model which is used to model ferroelectrics in external
electromagnetic fields \cite{Bax} and are representatives of the two classes studied e.g. in \cite{PR} in the theory of quantum integrable systems, \cite{HWKK}, going back to \cite{Wu}. The pair $(R,L)$  is called in \cite{Korff} (observing the switch in notation) the {\it osculating walkers} model, and the pair $(R',L')$ the {\it vicious walkers} model.
\normalsize

\begin{remark}  
\label{fivevertex}
 It is known that a (generic) pair $(R(x,y), L(x,t))$ of matrices, each of
the form \eqref{6vertex}, is a solution of \eqref{RLL}, if the Boltzmann weights $(\omega_1,\omega_2,\omega_3,\omega_4,\omega_5,\omega_6)$  for each matrix yield constant values for the following two ratios, that is 
\begin{eqnarray*}
\Delta_1 = \frac{ \omega_1\omega_2 + \omega_3\omega_4 - \omega_5\omega_6}{
2\omega_1\omega_3},\quad\text{and}\quad
\Delta_2 = \frac{\omega_1\omega_2 }{
\omega_1\omega_3}
\end{eqnarray*}
are well-defined complex numbers for either of the two matrices. This criterion is originally due to Baxter \cite{Bax}, but can be also found in \cite[Theorem 2]{BBF}. Note that this criterion fails for the (degenerate) vicious walker model.   
\end{remark}

\begin{remark}
\label{physics}
There are two types of transformations one can perform on the weights of a solution of \eqref{RLL} which produces a new solution, namely
\begin{eqnarray}
(\omega_1,\omega_2,\omega_3,\omega_4,\omega_5,\omega_6)&\leadsto& (\omega_2,\omega_1,\omega_4,\omega_3,\omega_6,\omega_5)\\
(\omega_1,\omega_2,\omega_3,\omega_4,\omega_5,\omega_6)&\leadsto& (\omega_3,\omega_4,\omega_1,\omega_2,\omega_6,\omega_5) \label{secondtype}
\end{eqnarray}
This follows easily diagrammatically, since the first exchanges $0$ and $1$ labels on the vertical and the horizontal edges of each vertex configuration, whereas the second exchanges $0$ and $1$ on the vertical edge and swaps the entries on the horizontal edge.
\end{remark}
 
\begin{proposition} \label{yangbaxter} The pairs $(R(x,y),L(x,t))$ and $(R'(x,y),L'(x,t))$ are both  solutions to the Yang-Baxter equation \eqref{RLL}.
\end{proposition}

\begin{proof}
For $(R',L')$ this follows directly from Remark~\ref{fivevertex}, as we find $\Delta_2 = 0$ in these cases.  Applying the transformation \eqref{secondtype} to $(R',L')$ and replacing $t$ by $-t$ we obtain the other pair $(R,L)$. By Remark~\ref{physics} it solves the Yang-Baxter equation as well. 
\end{proof}

\section{Yang-Baxter algebras and Bethe vectors}
\label{secYBA}
From now on we work with the solutions $(R,L)$ and $(R',L')$ from \eqref{L1L2} to the Yang-Baxter equation. We define now the main object of our studies. 

\subsection{Yang-Baxter algebras}
Let $O=O(x,t_1,\ldots, t_N)\in\{A,B,C,D\}$. Expand it as
\begin{eqnarray}
\label{ABCD}
O(x)=\sum_{i\geq 0} O^{(i)}x^i
\end{eqnarray}
in the variable $x$, with coefficients in $\op{End}_{\bP}(\mV_N)=\op{End}_\bP(V^{\otimes N}[t_1,t_2,\ldots, t_N])$.

\begin{proposition} 
\label{PropYBA}
In case of  $(R,L)$, the identity from Proposition~\ref{RMM}, 
\begin{eqnarray}
\label{YBEconcrete}
R_{12}(x_1,x_2)M_{1}(x_1,t)M_{2}(x_2,t)&=&M_{2}(x_2,t)M_{1}(x_1,t)R_{12}(x_1,x_2)
\end{eqnarray}
is equivalent to the following sixteen relations 
\begin{eqnarray*}
A(x_1)A(x_2)\;=\;A(x_2)A(x_1), && B(x_1)B(x_2)\;=\;B(x_2)B(x_1),\\
B(x_1)A(x_2)\;=\;B(x_2)A(x_1), && A(x_1)C(x_2)\;=\;A(x_2)C(x_1),\\
C(x_1)C(x_2)\;=\;C(x_2)C(x_1), && D(x_1)D(x_2)\;=\;D(x_2)D(x_1),\\
B(x_2)D(x_1)\;=\;B(x_1)D(x_2),&&D(x_2)C(x_1)\;=\;D(x_1)C(x_2),\\
B(x_1)C(x_2)\;=\;B(x_2)C(x_1), && B(x_1)A(x_2)\;=\;B(x_2)A(x_1),\\
A(x_2)B(x_1)- A(x_1)B(x_2)&=&(x_2-x_1)B(x_2)A(x_1), \\
C(x_2)A(x_1)-C(x_1)A(x_2)&=&(x_1-x_2)A(x_1)C(x_2),\\
A(x_1)D(x_2)-A(x_2)D(x_1)&=&(x_1-x_2)B(x_2)C(x_1), \\
D(x_2)B(x_1)-D(x_1)B(x_2)& =&(x_1-x_2)B(x_1)D(x_2),\\
C(x_1)D(x_2)-C(x_2)D(x_1)&=&(x_1-x_2)D(x_2)C(x_1),\\
C(x_2)B(x_1)-C(x_1)B(x_2)&=&(x_1-x_2)(A(x_1)D(x_2)-D(x_2)A(x_1)).
\end{eqnarray*}
\end{proposition}
 A similar reformulation of the Yang-Baxter equations can be given for  $(R',L')$.
\begin{proof} We will illustrate how to obtain the claim, but omit the straight-forward calculations. It suffices to assume $N=3$. We have to consider all possible inputs in \eqref{PropYBA} and compare their outputs. Let us consider  the input $(0,1,\la_1,\la_2,\la_3)$, with $\la_i\in\{0,1\}$. 

\begin{equation}
\label{LHSRHS}
\begin{tikzpicture}[scale=0.8]
\draw (1,0) node{=};
\draw[very thick, rounded corners=8pt] (0,0.5) -- (-2,0.5) -- (-3,-0.5) node[anchor=east]{0};
\draw[very thick, rounded corners=8pt] (0,-0.5) -- (-2,-0.5) -- (-3,0.5) node[anchor=east]{1};
\draw[very thick] (-0.3,-1) -- (-0.3,1) node[anchor=south]{$v_{\la_3}$};
\draw[very thick] (-0.9,-1) -- (-0.9,1) node[anchor=south]{$v_{\la_2}$};
\draw[very thick] (-1.5,-1) -- (-1.5,1) node[anchor=south]{$v_{\la_1}$};
\end{tikzpicture}
\quad
\begin{tikzpicture}[scale=0.8]
\draw[very thick, rounded corners=8pt] (0,0.5) -- (-1,-0.5) -- (-3,-0.5) node[anchor=east]{0};
\draw[very thick, rounded corners=8pt] (0,-0.5) -- (-1,0.5) -- (-3,0.5) node[anchor=east]{1};
\draw[very thick] (-2.7,-1) -- (-2.7,1) node[anchor=south]{$v_{\la_1}$};
\draw[very thick] (-2.1,-1) -- (-2.1,1) node[anchor=south]{$v_{\la_2}$};
\draw[very thick] (-1.5,-1) -- (-1.5,1) node[anchor=south]{$v_{\la_3}$};
\end{tikzpicture}
\end{equation}

Then we obtain with  ${\bf v}=v_\la$ for the diagram on the right hand side of  \eqref{LHSRHS} 
\begin{eqnarray*}
&&
R_{12}(x_1,x_2)(v_0\otimes v_0\otimes A(x_1)B(x_2){\bf v}+v_0\otimes v_1\otimes A(x_1)D(x_2){\bf v}\\
&&+
v_1\otimes v_0\otimes C(x_1)B(x_2){\bf v}+v_1\otimes v_1\otimes C(x_1)D(x_2){\bf v})\\
&=&
v_0\otimes v_0\otimes A(x_1)B(x_2){\bf v}+v_0\otimes v_1\otimes C(x_1)B(x_2){\bf v}+v_1\otimes v_0\otimes A(x_1)D(x_2){\bf v}\\
&&+v_0\otimes v_1\otimes (x_1-x_2)A(x_1)D(x_2){\bf v}+v_1\otimes v_1\otimes C(x_1)D(x_2){\bf v}
\end{eqnarray*}
The output of the left hand side of  \eqref{LHSRHS}  is
\begin{eqnarray*}
&&(x_1-x_2)(v_0\otimes v_0\otimes B(x_2)A(x_1){\bf v}+v_0\otimes v_1\otimes D(x_2)A(x_1){\bf v}\\
&+&
v_1\otimes v_0\otimes B(x_2)C(x_1){\bf v}+v_1\otimes v_1\otimes D(x_2)C(x_1){\bf v})\\
&+&
v_0\otimes v_0\otimes A(x_2)B(x_1){\bf v}+
v_0\otimes v_1\otimes C(x_2)B(x_1){\bf v}\\
&+&v_1\otimes v_0\otimes A(x_2)D(x_1){\bf v}
+v_1\otimes v_1\otimes C(x_2)D(x_1){\bf v}
\end{eqnarray*}
Comparing the coefficients of the basis vectors we obtain
\begin{eqnarray*}
A(x_2)B(x_1)- A(x_1)B(x_2)&=&(x_2-x_1)B(x_2)A(x_1),\\
C(x_1)D(x_2)-C(x_2)D(x_1)&=&(x_1-x_2)D(x_2)C(x_1),\\
C(x_2)B(x_1)-C(x_1)B(x_2)&=&(x_1-x_2)(A(x_1)D(x_2)-D(x_2)A(x_1)),\\
A(x_1)D(x_2)-A(x_2)D(x_1)&=&(x_1-x_2)B(x_2)C(x_1).
\end{eqnarray*}
These arguments applied to all possible inputs give the asserted 16 relations.
\end{proof}

Note that in fact each of the 16 equations is a system of equations for the operators $\eqref{ABCD}$. Following  \cite{FRT} we use them to define abstract algebras:

\begin{definition} 
\label{DefYBA}
The {\it Yang-Baxter algebra} $\op{YB}$  (or $\op{YB}'$) is the ${\bP}$-algebra generated by $\hat{A}^{(i)}$, $\hat{B}^{(i)}$, $\hat{C}^{(i)}$, $\hat{D}^{(i)}$ (respectively by $\hat{A'}^{(i)}$, $\hat{B'}^{(i)}$, $\hat{C'}^{(i)}$, $\hat{D'}^{(i)}$), $i\in\mathbb{Z}_{\geq0}$ modulo the relations (see Proposition~\ref{PropYBA}) given by the Yang-Baxter equation \eqref{YBEconcrete} for $(R,L)$ respectively $(R',L')$ from \eqref{L1L2}.
\end{definition}

These algebras act on $\mV_N$, by sending the generators to the endomorphisms 
\begin{eqnarray}
\label{YBABCD}
A^{(i)}_N=A^{(i)}_N(t_1,t_2,\ldots, t_N),&&B^{(i)}_N=B^{(i)}_N(t_1,t_2,\ldots, t_N),\nonumber\\
C^{(i)}_N=C^{(i)}_N(t_1,t_2,\ldots, t_N),&&D^{(i)}_N=D^{(i)}_N(t_1,t_2,\ldots, t_N).
\end{eqnarray}
from \eqref{ABCD} for all $i\geq 0$. In the following we will often omit the lower index $N$. 

Note that the space $V^{\otimes{N}}$ has a natural decomposition $V^{\otimes{N}}= \bigoplus_{i=0}^N V_n$ into subspaces where $V_n$ is spanned by all standard basis vectors 
 $v_{\la}$
from \eqref{basisvec} such that $\la\in\Lambda_n$. Similarly, we have 
 $\mV_N=\bigoplus_{n=0}^N\mV_{N,n}$, the corresponding decomposition into free $\bP$-modules, which we call {\it weight spaces}. The action of $\op{YB}$ is compatible with the weight space decomposition in the following sense.

\begin{lemma}  
\label{weightspaces}  Consider the monodromy matrix for  $(R,L)$ and $(R',L')$ as in \eqref{L1L2}.
For every $i\geq 0$ and $0\leq n\leq N$ it holds
\begin{eqnarray*}
\quad A^{(i)}_N:\mV_{N,n}\rightarrow \mV_{N,n}, \quad B^{(i)}_N:\mV_{N,n}\rightarrow \mV_{N,n+1},\quad
C^{(i)}_N:\mV_{N,n}\rightarrow \mV_{N,n-1}, \quad D^{(i)}_N:\mV_{N,n}\rightarrow \mV_{N,n}.
\end{eqnarray*}
\end{lemma}
\begin{proof} We only argue for $(R,L)$, the other case is similar. For $N=1$ the matrix $M_1(x,t_1)$  is the Lax matrix, hence the generators of the algebra $\op{YB}_1$ are the operators given by the coefficients of the expansion in $x$ of the $2\times 2$ blocks of $L(x,t)$, explicitly 
\begin{eqnarray*}
A^{(0)}=\begin{pmatrix}
1&0\\
0&t
\end{pmatrix},\quad
A^{(1)}=\begin{pmatrix}
0&0\\
0&1
\end{pmatrix},
\quad
B^{(0)}=
\begin{pmatrix}
0&0\\
1&0
\end{pmatrix},
\quad
C^{(0)}=
\begin{pmatrix}
0&1\\
0&0
\end{pmatrix},
\quad
D^{(0)}=
\begin{pmatrix}
0&0\\
0&1
\end{pmatrix},
\end{eqnarray*}
which obviously have the asserted property. The definition of the monodromy matrix \eqref{mono0} implies immediately the claim for any $N$. 
\end{proof}

\subsection{Symmetric group action and $\bH$-action}
\label{symmetricgroup}
The symmetric group $S_N$  acts on $\mV_N=V\otimes\bP$, such that $w\in S_N$ sends $v_{\la}\otimes f$  to $w(v_{\la})\otimes w(f)=v_{\la'}\otimes w(f)$ with $\la'_j=\la_{w^{-1}(j)}$,  that is, $w$ simultaneously commutes the indices of the basis vector and the variables of the polynomials.\footnote{Some readers might prefer to define ${\bf s}_i=\sigma R(t_i,t_{i+1}) P_{i,i+1}$, where $P_{i,i+1}$ is the $\bP$-linear flip operator on $V^{\otimes N}$ and $\sigma$ is the permutation of the variables $t_i$ and $t_{i+1}$.}
Since our $R$-matrices \eqref{L1L2} satisfy  $s_1 R(t_2,t_1) s_1 R(t_2,t_1)=\op{id}_{\mV_2}$, the operators 
\begin{eqnarray}
\label{defsi}
{\bf s}_i&=&s_i R_{i,i+1}(t_{i+1},t_i) 
\end{eqnarray}
for $i\in [N-1]$ give rise to an action of the symmetric group $S_N$ on $\mV_N$. This action is not $\bP$-linear, but part of an action of the {\it semidirect product algebra} $\bH$ (or {\it skew group ring}) which is defined as follows, where we from now on write ${}^w\! f$ instead of $w(f)$ when $w\in S_N$ is applied to $f\in\bP$. 
\begin{definition}
\label{DefH}
The algebra $\bH$ is the vector space $\bH= \Bbb C[S_N]\otimes\bP$
with the multiplication such that  $\Bbb C[S_N]$ and ${\bP}$ are subalgebras and $fw=  {}^{w^{-1}}\!\!f w$ for $w\in S_N, f\in {\bP}$.
\end{definition}
Explicitly, $f\in \bP$ acts by multiplication with $f$, and the action of ${\bf s}_i$ is given as follows.
\begin{proposition} 
\label{Haction}
On the standard $\bP$-basis vectors $v_{\la}$ we have 
\begin{eqnarray}
\label{sym}
{\bf s}_i(v_{\la})
&=&
\begin{cases}
v_{\la}+(t_i-t_{i+1})v_{s_i\la},& \text{if}\,\la_i<\la_{i+1},\\
v_{\la},&\text{otherwise}.
\end{cases}
\end{eqnarray}
which then extends to the desired action of $\bH$ by $s_i(fv_{\la})=  {}^{s_i}\!f\; {\bf s}_i(v_{\la})$.
\end{proposition}
\begin{remark}
Note that the result of $\op{id}-{\bf s}_i$ is divisible by $t_i-t_{i+1}$, hence the quotient is well-defined. It is the well-known Demazure operator $\Delta_i$, \cite{Demazure}, a standard tool, see e.g. \cite{Kumar}, \cite{KT03},  in cohomology theory, see also Section~\ref{eqSchubert}.
\end{remark}

Our Yang-Baxter algebras have the important property, pointed out already in \cite{GoKo}, that they commute with this $S_N$-action. More precisely the following holds.
\begin{proposition}
\label{Propcommute}
Let $N\geq 2$ and $i\in [N-1]$. Then  there is an equality  
\begin{eqnarray*}
(1\otimes {\bf s}_i)M(x,t_1,t_2,\ldots, t_N)=M(x,t_1,t_2,\ldots, t_N)(1\otimes {\bf s}_i),
\end{eqnarray*}
of endomorphisms of $V[x]\otimes \mV_N$, with  monodromy matrix $M$ for $L$ or $L'$ from \eqref{L1L2}.
\end{proposition}
\begin{proof} 
The statement is local and involves only three tensor factors, namely the leftmost factor at position $0$ and the factors $i$ and $i+1$ in  $\mV_N$. We may assume $i=1$. 
First consider the Yang-Baxter equation 
\begin{eqnarray*}
L_{12}(x_2,t)L_{02}(x_1,t)R_{01}(x_1,x_2)&=&R_{01}(x_1,x_2)L_{02}(x_1,t)L_{12}(x_2,t)
\end{eqnarray*}
Substituting $x_2=-t_1$, $t=t_2$ we obtain
\begin{eqnarray*}
R_{12}(t_2,t_1)L_{02}(x_1,t_2)L_{01}(x_1,t_1)&=&L_{01}(x_1,t_1)L_{02}(x_1,t_2)R_{12}(t_2,t_1)
\end{eqnarray*}
using the special shape of $L$ and $R$, e.g. $L(-x,y)=R(y,x)$.
Thus
\begin{eqnarray*}
(1\otimes {\bf s}_1)M(x,t_1)
&=&(1\otimes {\bf s}_1)L_{02}(x,t_2)L_{01}(x,t_1)\\
&\stackrel{\eqref{defsi}}{=}&(1\otimes {s}_1)R_{12}(t_2,t_1)L_{02}(x,t_2)L_{01}(x,t_1)\\
&=&(1\otimes {s}_1)L_{01}(x,t_1)L_{02}(x,t_2)R_{12}(t_2,t_1)\\
&=&L_{02}(x,t_2)L_{01}(x,t_1)(1\otimes {s}_1)R_{12}(t_2,t_1)\\
&=&M(x,t_1)(1\otimes {\bf s}_1).
\end{eqnarray*}
For the pair $(R',L')$ the following relation can be verified by a direct calculation
\begin{eqnarray*}
R'_{12}(t_2,t_1)L'_{02}(x_1,t_2)L'_{01}(x_1,t_1)&=&L'_{01}(x_1,t_1)L'_{02}(x_1,t_2)R'_{12}(t_2,t_1)
\end{eqnarray*}
Since the symmetric group action is defined as ${\bf s}_i=s_i R'_{i,i+1}(t_{i+1},t_i)$, 
the proof that it commutes with the monodromy matrix $M'$ is the same as above.
\end{proof}

\subsection{Bethe vectors}
The Yang-Baxter algebra contains several interesting commutative subalgebras. In the theory of quantum integrable systems there is a procedure, called the {\it (algebraic) Bethe Ansatz} or the {\it quantum inverse scattering method}, \cite{quantumscatt}, \cite{BIK}, for finding a common eigenbasis. In general this is a rather sophisticated procedure, see e.g. \cite{Tak} for a general introduction, but for the purposes of this paper it suffices to consider the simplest version of it.
 
 The relation $A(x_1)A(x_2)=A(x_2)A(x_1)$ in the Yang-Baxter algebra means that the coefficients of the operator $A(x)$ commute, and hence can be diagonalised simultaneously. The key relation  for doing this is
\begin{eqnarray}
\label{CA}
C(x_2)A(x_1)-C(x_1)A(x_2)&=& (x_1-x_2)A(x_1)C(x_2)
\end{eqnarray}

\begin{lemma} 
\label{lemmaAC}
Let $y = (x,y_1,\ldots,y_k)$ be a $k+1$-tuple of pairwise distinct commuting variables. Then we have the following formal identities:
$$ 
A(x)C(y_k)\cdots C(y_1) = \frac{C(y_k)\cdots C(y_1)A(x)}{(x-y_k)\cdots (x-y_1)}-\sum_{i=1}^k\frac{C(y_k)\cdots \overset{i}{C(x)}\cdots C(y_1)A(y_i)}{(x -y_i)\prod_{i\not=j}(y_i-y_j)},
$$
where the superscript $i$ indicates the $i$th factor.
\end{lemma}
\begin{proof} We do induction using \eqref{CA} with starting point \eqref{CA} itself. Suppose 
$$A(x)C(y_{k})\cdots C(y_2) = \frac{C(y_{k})\cdots C(y_2)A(x)}{(x-y_{k})\cdots (x-y_2)}-\sum_{i=2}^{k}\frac{C(y_{k})\cdots \overset{i}{C(x)}\cdots C(y_2)A(y_i)}{(x -y_i)\prod_{i\not=j}(y_i-y_j)}$$
holds. Then $A(x)C(y_k)\cdots C(y_1)$ equals 
\begin{eqnarray*}
&&\frac{C(y_k)\cdots C(x_2)A(x)C(y_1)}{(x-y_k)\cdots (x-y_2)}-\sum_{i=2}^{k}\frac{C(y_{k})\cdots \overset{i}{C(x)}\cdots C(y_2)A(y_i)C(y_1)}{(x -y_i)\prod_{i\not=j}(y_i-y_j)}\\
&=&\frac{C(y_k)\cdots C(y_2)(C(y_1)A(x)-C(x)A(y_1))}{(x-y_k)\cdots (x-y_1)}\\
&&-\sum_{i=2}^{k}\frac{C(y_{k})\cdots \overset{i}{C(x)}\cdots C(y_2)(C(y_1)A(y_i)-C(y_i)A(y_1))}{(x -y_i)(y_i-y_1)\prod_{i\not=j}(y_i-y_j)}
\end{eqnarray*}
Hence the proof will be finished if we show that 
\begin{equation}
\label{firstInd}
\frac{C(y_k)\cdots C(y_2)C(y_1)A(x)}{(x-y_k)\cdots (x-y_1)}-\sum_{i=2}^{k}\frac{C(y_{k})\cdots \overset{i}{C(x)}\cdots C(y_2)C(y_i)A(y_1)}{(x -y_i)(y_i-y_1)\prod_{i\not=j}(y_i-y_j)}=\frac{C(y_{k})\cdots C(y_2)C(x)A(y_1)}{(x -y_1)\prod_{j\not=1}(y_1-y_j)}.
\end{equation}
Since the operators $C(x)$ and the $C(y_i)$'s pairwise commute, the numerators on both sides coincide. Therefore the claim is reduced to the equality
\small
$$\frac{1}{(x-y_k)\cdots (x-y_1)}-\sum_{i=2}^{k}\frac{1}{(x -y_i)(y_i-y_1)\prod_{i\not=j}(y_i-y_j)}=\frac{1}{(x -y_1)\prod_{j\not=1}(y_1-y_j)}$$
\normalsize
which can easily be verified.
\end{proof}
The Bethe Ansatz is a method of finding eigenvectors  for the $A^{(i)}_N$, $i\geq 0$ from \eqref{YBABCD} of the form 
\begin{eqnarray}
\label{BetheDef}
\bb(\eta)=C(y_k)\cdots  C(y_1)(v_1\otimes \cdots \otimes v_1)_{|_{y_i=\eta_i}}\in \mV_{N,n}
\end{eqnarray}
depending on tuples $\eta=(\eta_1,\ldots\eta_k)\in\bP^k$. The parameters are solution to a system of algebraic equations called the {\it Bethe Ansatz equations} which are hard to compute in general.  However in our case the solutions are easy to find using Lemma~\ref{lemmaAC} and the identity 
\begin{equation}
\label{Bethecalc}
A(x)v_{(1,1,\dots,1)}=\prod_{i=1}^N(x+t_i)v_{(1,1,\dots,1)}.
\end{equation}
Note that permuting the components of $\eta$ in \eqref{BetheDef} provides the same vector, since the operators $C(y_{i})$ pairwise commute. 

\begin{proposition} For $\eta=(-t_{i_1},\ldots ,-t_{i_k})$ where $I_0=\{i_j\mid 1\leq j\leq k\}\subseteq [N]$ with distinct elements, 
\begin{eqnarray}
\label{Bethevector}
\bb(\eta)&=&C(-t_{i_1})\cdots C(-t_{i_k})v_{(1,1,\dots,1)}\in\mV_N
\end{eqnarray}
is a simultaneous  eigenvector for the $A^{(i)}_N$'s with eigenvalues given by  
\begin{eqnarray}
\label{BetheEV}
A(x)\bb(\eta)&=&\displaystyle{\prod_{j\notin \{i_1,\cdots ,i_k\}}(x+t_j)}\bb(\eta).
\end{eqnarray} 
The $\bb(\eta)$ for ${i_1}<i_2<\cdots <{i_k}$ form a basis of $\mV_{N,n}$, the \emph{basis of Bethe vectors}.
\end{proposition}
\begin{proof}
Indeed setting $y_j=-t_{i_j}$ for $j=1,\cdots ,k$ in the statement of Lemma~\ref{lemmaAC} and applying the operators to the vector $v_{(1,1,\dots,1)}$, we see that on the RHS only the first summand survives due to \eqref{Bethecalc}. Then \eqref{BetheEV} follows.
\end{proof}

\begin{example} Let $k=1$. Then  a typical diagram representing a matrix element of the operator $C(x)$ acting on the vector $v_{(1,\ldots ,1)}$ looks as follows

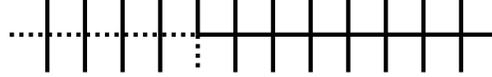
\[
\begin{tikzpicture}
\draw[ultra thick, dotted] (0,0) -- (2.5,0); 
\draw[ultra thick] (2.5,0) -- (6.5,0);
\draw[ultra thick] (2.5,0) -- (2.5,0.5);
\draw[ultra thick, dotted] (2.5,0) -- (2.5,-0.5);
\foreach \x in {0.5,1,...,2} \draw[ultra thick] (\x,-0.5) -- (\x,0.5);
\foreach \x in {3,3.5,...,6} \draw[ultra thick] (\x,-0.5) -- (\x,0.5);
\end{tikzpicture}
\]

\noindent
and more precisely, we obtain the formula
$$C(x)v_{(1,\ldots,1)}=\sum_{i=1}^N \prod_{j=1}^{i-1}(x+t_j)v_{\la(i)}\quad \text{where}\quad {\la(i)}=(1,\ldots,1,\overset{i}{0},1,\ldots,1),$$
and via evaluation at  $x=-t_{i_1}$ for $1\leq i_1\leq N$,  the corresponding Bethe vector 
$$\bb(\eta)=\bb((i_1))=\sum_{i=1}^N\prod_{j=1}^{i-1}(t_j-t_{i_1})v_{\la(i)}$$
expressed in the standard basis. In particular $\bb(0,1,\ldots,1)=v_{(0,1,\ldots,1)}$.
\end{example}

\begin{lemma}\label{hvector}
For $\la$ of the form $(0,\ldots,0,1,\ldots,1)$ the corresponding Bethe vector equals the standard basis vector, $\bb_{(0,\ldots,0,1,\ldots,1)}=v_{(0,\ldots,0,1,\ldots,1)}$. 
\end{lemma}
\begin{proof}
The only diagram contributing to the Bethe vector \eqref{Bethevector} indexed by $(0,\ldots,0,1,\ldots,1)$ is

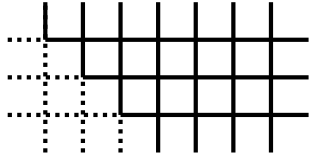
\[
\begin{tikzpicture}
\draw[ultra thick] (0.5,0) -- (4,0);
\draw[ultra thick] (1,-0.5) -- (4,-0.5);
\draw[ultra thick] (1.5,-1) -- (4,-1);
\draw[ultra thick, dotted] (0,0) -- (0.5,0);
\draw[ultra thick,dotted] (0,-0.5) -- (1,-0.5);
\draw[ultra thick,dotted] (0,-1) -- (1.5,-1);
\draw[ultra thick] (0.5,0) -- (0.5,0.5);
\draw[ultra thick,dotted] (0.5,-1.5) -- (0.5,0.5);
\draw[ultra thick] (1,-0.5) -- (1,0.5);
\draw[ultra thick,dotted] (1,-1.5) -- (1,-0.5);
\draw[ultra thick] (1.5,-1) -- (1.5,0.5);
\draw[ultra thick,dotted] (1.5,-1.5) -- (1.5,-1);
\foreach \x in {2,2.5,...,3.5} \draw[ultra thick] (\x,-1.5) -- (\x,0.5);
\end{tikzpicture}
\]

Since the product of the weights in this diagram is equal to $1$, we are done.
\end{proof}

 We need the localisation of $\bP$ and $\mV_N$ with respect to the multiplicative set $\op{loc}$ given by products of the form $(t_a-t_b)$, where $a,b\in[N]$,  
 \begin{eqnarray}
 \label{loc}
\bP^{\op{loc}}=\bP[f^{-1}\mid f\in \op{loc}]&\text{and}&\mV_N^{\op{loc}}=V^{\otimes N}[t_1,t_2,\ldots, t_N][f^{-1}\mid f\in \op{loc}],
\end{eqnarray}
 to make sense of this normalisation (giving rise to normed Bethe vectors for an appropriate form).
\begin{definition}
\label{normalizedBethe}
Given  ${\bf \la}\in\Lambda$ let $I_0=\{i\mid \la_i=0\}\subset [N]$. Then 
the associated (normalised) Bethe vector $\bb_{\la}\in\mV_N^{\op{loc}}$ is the vector $\bb(\eta)$ from \eqref{Bethevector}  normalised as 
\begin{eqnarray*}
\bb_\la&=&
\prod_{b\in I_0,\,a\notin I_0}(t_a-t_b)^{-1}\bb(\eta)\quad\in \mV_N^{\op{loc}}.
\end{eqnarray*}
The \emph{(normalised) Bethe basis} of $\mV_N^{\op{loc}}$ is the $\bP^{\op{loc}}$-basis given by $\bb_\la$, $\la\in\Lambda$.
\end{definition}

The $S_N$-action  \eqref{defsi}  extends to $\mV_N^{\op{loc}}$ and just permutes the Bethe basis vectors:
\begin{proposition} \label{SNonBethe}
The action of $S_N$ from \eqref{defsi} satisfies for $j\in [N-1]$
\begin{eqnarray}
{\bf s}_j(\bb{(\la)})=\bb({s_j (\la}) )\quad\text{resp.}\quad {\bf s}_j (\bb_{\la})=\bb_{s_j\la},\quad\text{ where $\la\in\Lambda$}.
\end{eqnarray}
\end{proposition}
\begin{proof} 
The $S_N$-action from \eqref{defsi} commutes with the Yang-Baxter operators by Proposition~\ref{Propcommute} and ${\bf s}_j( v_{(1,\ldots,1)})=v_{(1,\ldots,1)}$ by the definition of the $R$-matrix. Therefore, acting by ${\bf s}_j$ on the vector $C(-t_{i_k})\ldots C(-t_{i_1})v_{(1,\ldots,1)}$ amounts to applying the permutation $s_j$ to the $-t_{i_j}$'s and the normalisation factor. 
\end{proof}

\begin{remark}
For the pair $(R',L')$, the arguments are completely analogous. Instead of \eqref{CA} one would consider the equation $B'(x_1)A'(x_2)-B'(x_2)A'(x_1)=(x_2-x_1)A'(x_1)B'(x_2)$.  Then the vectors $\bb'(\eta)=B'(t_{i_1})\cdots B'(t_{i_n})v_{(0,0,\dots,0)}\in\mV_N$  defined similar to  \eqref{Bethevector} satisfy the equality $A'(x)\bb'(\eta)=\displaystyle{\prod_{j\notin \{i_1,\cdots ,i_n\}}(x-t_j)}\bb(\eta)$. In fact, $\bb'(\eta')=\bb(\eta)$, where $\eta'$ is complementary to $\eta$. 
\end{remark}

\begin{remark}
The readers who are surprised about the fact that only the operators $A(x)$, and not the $D(x)$, occur in the construction of the Yang-Baxter algebra are referred to Section~\ref{secquantum}.
\end{remark}
\section{The $R$-matrix from the geometry of  Grassmannians}
We first recall some facts from the Schubert calculus for Grassmannians, see e.g. \cite{Brion}, \cite{Fulton} for more details. Then we will interpret  $\mV_N$ geometrically as $\bigoplus_{n=0}^N H_T^*(\Xn)$ identifying the standard basis with an equivariant Schubert basis.  After some localisation, the normalised Bethe basis makes sense. It will correspond under the localised isomorphism to the easiest possible basis in geometry
\begin{eqnarray*}
\mV_N^{\op{loc}}&\cong& \bigoplus_{n=0}^N H_T^*(\Xn)^{\op{loc}}\\
\text{normalised Bethe basis} &\leftrightarrow& \text{$T$-fixed point basis}
\end{eqnarray*}
namely the fixed point basis for the $T$-action, as we show in Corollary~\ref{Bethe}.
\subsection{Basics}
\label{Schubertclasses}
We consider now $\mathbb{C}^N$, fix its standard basis $e_1,\ldots, e_N$ and let $G=\op{GL}_N(\Bbb C)$. For $0 \leq n \leq N$ we denote by $\Xn=\op{Gr}(n,N)$ the corresponding {\it Grassmannian variety} of $n$-dimensional subspaces in $\Bbb C^N$. It is the projective variety $\op{Gr}(n,N) = G/P$ where $P=P_{n,k}$ is the standard parabolic subgroup consisting of the block upper triangular invertible matrices with diagonal blocks of sizes $n$ and and $k:=N-n$ containing the standard Borel subgroup $B$ of all upper triangular matrices. Let $T\subset G$ be the torus of diagonal matrices. It acts on $\Xn$ because it acts obviously on $\Bbb C^N$, respectively by left multiplication on $G/P$.
The set $\Xn^T$ of fixed points of this action is finite. Explicitly, there is a bijection
$\Lambda_n \stackrel{1:1}{ \leftrightarrow} \Xn^T$, $\lambda\mapsto  \Bbb C^{\lambda}$,
sending $\lambda\in \Lambda_n$ to the coordinate plane $\Bbb C^{\lambda}$ spanned by all $e_i$ with $\lambda_i=1$. Fix the {\it standard flag} 
\begin{eqnarray*}
\op{F}^e&=&(\{0\}=F_0\subset F_1\subset F_2\subset\ldots\subset F_N=\mC^N), 
\end{eqnarray*}
where $F_i$ is spanned by $e_j$ for $j\in [i]$. For $\lambda\in \Lambda_n$ let $C_{\lambda}$ be the {\it Schubert cell} in $\Xn$,
\begin{eqnarray}
\C_{\lambda}&=& \left\{V \in \Xn \mid\lambda_i=\op{dim}\left((V \cap F_{i})/ (V \cap F_{i-1})\right), \forall\; 1\leq i\leq N\right\},\quad\quad
\end{eqnarray}
and let $\Omega_{\lambda}$  be its {\it Schubert variety}, i.e.\  the closure of the cell $\C_{\lambda}\subset \Xn$. Note that $\C_{\lambda}$ contains exactly one $T$-fixed point, namely $\mC^\lambda$. Define the {\it inversion set} of  $\lambda\in \Lambda_n$ as 
\begin{eqnarray*}
\op{inv} (\lambda)&=&\{(i,j)\mid 1\leq i<j\leq N, \lambda_i>\lambda_j\}
\end{eqnarray*}
and let $\ell(\lambda)=|\op{inv} (\lambda)|$, e.g.  $\ell((0,1,1))=0$, $\ell((1,0,1))=1$ and $\ell(1,1,0)=2$.

The {\it Schubert class} $\oS_{\lambda} \in H^*(\Xn)=H^*(\Xn,\mC)$ is the Poincare dual in cohomology of the cycle in homology represented by the Schubert variety $\Omega_{\lambda}$. In particular, the degree of $\oS_{\lambda}$ is $2\ell(\lambda)$. These classes are well known to form a basis (even over $\mathbb{Z}$) of $H^*(\Xn)$.
Since the cells $C_{\lambda}$ are $T$-invariant (in fact $B$-orbits) they define also $T$-equivariant cohomology classes which we denote by abuse of notation also by $\oS_{\lambda}$. They form a $H_T^*(\op{pt})$-basis for the $T$-equivariant cohomology $H_T^*(\Xn)$. Later we will need a twisted version of this equivariant Schubert basis which we introduce now. 

\subsection{Twisted Schubert classes}
Let $w\in S_N$. Acting with $w$ on the standard flag $\op{F}^e$ produces a new flag $\op{F}^{w}$, where $F^{w}_i$ is spanned by $e_{w (j)}$, $j\leq i$. Note that $\op{F}^e$ corresponds to the neutral element $e\in S_N$. Mimicking the construction of $C_{\lambda}$ for the flag $\op{F}^{w}$ instead of $\op{F}^{e}$, we obtain another cell decomposition with cells
\begin{eqnarray*}
\C^w_{\lambda}&=& \left\{V \in \Xn \mid\lambda_{w(i)}=\op{dim}\left((V \cap F^w_{i})/ (V \cap F^w_{i-1})\right) \text{ for all } i\in [N]\right\}.
\end{eqnarray*}
Note that $\C^w_{\lambda}$ contains again exactly one fixed point, namely $\mC^{w(\la)}$. The closure of $\C^w_{\lambda}$ is denoted $\Omega^{w}_{\lambda}$ and called the {\it $w$-twisted Schubert variety}.  
\begin{example}
\label{exP1}
Let us consider $\op{Gr}(1,2)=\Bbb P^1$. The $T$-fixed points are represented by coordinate lines $p_1=\mC e_1$ and $p_2=\mC e_2$ corresponding to the $\{0,1\}$-words $\la=(1,0)$ and $\la=(0,1)$ respectively. For the standard flag $\op{F}^e$ the Schubert cells are the single point, $C_{(1,0)}=\{p_1\}$, 
 and the \lq big\rq cell, $C_{(0,1)}$, containing the fixed point $p_2$. For the flag $\op{F}^{s_1}$ the cell $C^{s_1}_{(1,0)}$ is now the \lq big\rq\ cell containing the fixed point $p_1$, and $C^{s_1}_{(0,1)}$ is the single point $p_2$.
\end{example}

The Schubert cells are the $B$-orbits in $\Xn$, namely $\C_\la=B\mC^\la$.  Moreover, $V\in \C^w_{\lambda}$ 
if and only if $\op{dim}\left((w^{-1}(V )\cap F_{i})/ (w^{-1}(V) \cap F^w_{i-1})\right)=\la_{w(i)}=(w^{-1}(\la))_i$ for all $i$, that is $w^{-1}(V)\in B\mC^{w^{-1}(\la)}$ or equivalently $V\in wB\mC^{w^{-1}(\la)}=wBw^{-1}w\mC^{w^{-1}(\la)}$. Hence twisting with $w$ means taking the orbits for the conjugated Borel subgroup $B^w=wBw^{-1}$ in $\Xn$, in formulae
\begin{eqnarray}
\label{wBw}
\C_\la^w&=&wBw^{-1}\mC^{\la}. 
\end{eqnarray}

Another convenient way, is to work with attracting cells for a $\mC^*$-action. For this pick integers $a_1, a_2, \ldots, a_N$  and consider the  corresponding integral cocharacter $\tau^{\bf a}:\mC^*\rightarrow T, t\mapsto t^{\bf a}$, where  $t^{\bf a}$ is the matrix with diagonal entries $t^{a_1}, t^{a_2},\ldots, t^{a_n}$. 
Then there is the following alternative description:

\begin{proposition} 
Assume $a_1>a_2>\ldots>a_N$. 
\begin{enumerate}[1.)]
\item Then the cell  $\C_{\lambda}$ is the attracting manifold to the fixed point $p=\mC^\lambda$, i.e.\  $\C_{\lambda}=\{x\in \Xn\mid \lim_{t\rightarrow 0}\tau^{\bf a}(t).x=p\}.$ 
\item The cell $\C^{w}_{\lambda}$ is the attracting manifold to the fixed point $p=\mC^\lambda$ for the cocharacter $t^{w^{-1}(\bf a)}$ corresponding to $a_{w^{-1}(1)}, \ldots, a_{w^{-1}(N)}$.
\end{enumerate}
\end{proposition}
\begin{proof}
The first statement is clear and for the second note then that 
$w^{-1}(x)\in B\mC^{w^{-1}(\la)}$ if  and only if $\lim_{t\rightarrow 0}\tau^{\bf a}(t).w^{-1}(x)=\mC^{w^{-1}(\la)}$ and so the statement follows from $w\tau^{\bf a}(t)w^{-1}=\tau^{w^{-1}(\bf a)}(t)$.
\end{proof}
\begin{remark}
\label{weyl}
Note that if  we view ${\bf a}$ as an element of the integral weight lattice for $\mathfrak{gl}_N$, then the attached attracting cell is constant on Weyl chambers. 
\end{remark}

\subsection{Equivariant twisted Schubert classes}
\label{eqSchubert}
The cells $\C^{w}_{\lambda}$,  for fixed $w$ and varying $\la$, define also equivariant cohomology classes $S^{w}_{\lambda}$ which, as we will explain now, form a $H_T^*(pt)$-basis, and  $S^{e}_{\lambda}= S_{\lambda}$ will give the (untwisted) equivariant Schubert classes. 
For more details on equivariant cohomology we refer to  \cite{GKM}, \cite{Libine} and for our special case to \cite{KT03}. 

To describe the $T$-equivariant cohomology of $H_T^*(\Xn)=H_T^*(\Xn,\mathbb{C})$ of $\Xn$ we will use the $T$-equivariant map $\Xn^T\rightarrow \Xn$ which is known to induce a monomorphism $\op{res}: H_T^*(\Xn)\rightarrow H_T^*(\Xn^T)$ often called the {\it restriction to the fixed points}.

Let $\Func=\op{Maps}(\Lambda_n,\bP)$ be the set of functions on $\Lambda_n$ with values in $\bP$. We identify $\bP$ with $H^*_T(\op{pt})$, and  $\Func$ with $H^*_T(\Xn^T)$, and view $H^*_T( \Xn)\subset \Func$ via $\op{res}$. Then, by \cite{GKM}, see also \cite[Theorem 98]{Libine},  the  subset $H_T^*(\Xn)\subset \Func(n)$ can be identified with the set of  GKM-classes.

\begin{definition}
An element $\alpha \in \Func$ is called a {\it GKM class} or just {\it a class} if it satisfies the following Goresky-Kottwitz-MacPherson (GKM) conditions: 
\begin{eqnarray}
\label{GKMcond}
\alpha({\lambda}) - \alpha(\mu)&=&0\mod (t_i - t_j), \quad\text{if ${\mu} = s_{ij}{\lambda}$}, 
\end{eqnarray}
for a transposition $s_{ij}\in S_N$. For a class $\alpha$ define its {\it support} as $\Supp(\alpha)=\{\lambda\mid \alpha({\lambda})\not= 0\}\subset\Lambda_n$.
\end{definition} 

We now give a combinatorial description of the twisted equivariant Schubert classes, generalising the (ordinary) equivariant Schubert classes from  \cite{KT03}. 
\begin{lemma}
\label{Demazurelemma}
Consider the action of $S_N$ on $\Func$ in the usual way, namely
\begin{eqnarray}\label{symactioneq}
(w\cdot \alpha)(\lambda) &= &w\left(\alpha(w^{-1}\cdot\lambda)\right),
\end{eqnarray}
for $w\in S_N$, $\alpha\in\Func$. Then the following holds.
\begin{enumerate}[1.)]
\item If $\alpha$ is a class, then so is $w\cdot\alpha$ for any $w\in S_N$.
\item $\mu\in \Supp(\alpha)$ if and only if $w(\mu)\in \Supp(w\cdot\alpha)$.
\item  $H_T ^*(\Xn)\subset \Func$ is a subring and a $\bP$-submodule.
\item The Demazure operators $\Delta_i: f\mapsto \frac{f-{}^{s_i}\!f}{t_{i}-t_{i+1}}$,  $i\in [N-1]$, acting on $\Func$, preserve $H_T^*(\Xn)$.
\end{enumerate} 
\end{lemma}
\begin{proof}
Except of the last part this follows directly from the definitions. For the last part we have to verify the GKM-conditions for $\Delta_a(\alpha)$ for any class $\alpha$. We have $\Delta_a(\alpha)(\la)-\Delta_a(\alpha)(s_a\la)=0$ if $\alpha $ is a class, hence the condition \eqref{GKMcond} holds for $\{i,j\}=\{a,a+1\}$. Otherwise $t_i-t_j$ and $t_a-t_{a+1}$ are coprime in $\bP$. Since the difference $\alpha-s_a(\alpha)$ is a class, $\Delta_a(\alpha)$ is a class as well.
\end{proof}

Given  $w\in S_N$, define a partial ordering ${\geq_{w}}$ on $\Lambda_n$ as follows.
\begin{definition}   For $\lambda, \mu\in \Lambda_n$ let
\begin{eqnarray} 
\mu\geq_{w} \la \quad :\Leftrightarrow \quad w^{-1}(\mu)\geq w^{-1}(\la),&\text{with}& \mu\geq\lambda \quad\text{if}\quad  \sum_{i=1}^j  \mu_i \geq \sum_{i=1}^j \la_i\quad \text{for all} \quad j \in[N]. 
\end{eqnarray}
In particular $\geq_{\text {id}}\,=\,\geq$. We denote by $\Lambda_{\geq_{w} \lambda}$ all elements $\nu\in\Lambda_n$ such that  $\nu \geq_{w} \la$. 
\end{definition}

Since $\mC^\mu\in\Omega_\la^w$ if and only if $\mu \geq_{w} \la$, we are interested in the following classes.
\begin{definition}
Let $w\in S_N$ and $\la\in\Lambda_n$. We say a class $\alpha\in\Func$ {\it has $w$-support above $\lambda$} if $\la\in\Supp(\alpha)$ and $\Supp(\alpha)\subseteq \Lambda_{\geq_{w} \lambda}$.
\end{definition}

The GKM-conditions imply that if $\alpha$ has $w$-support above $\la$, then $\alpha(\la)$ is divisible by $\prod_{(i,j)\in inv_{w (\lambda)}}(t_i-t_j)$. We consider such classes of minimal total degree in the variables of $\bP$. 

\begin{lemma}
\label{unique}
Let $w\in S_N$, $\la\in\Lambda_n$. There exists a unique class $S^w_\lambda\in H^*_T (\Xn)$ satisfying
\begin{itemize}
\item [\textit{Stab1:}] $S^w_\lambda$ has $w$-support above $\la$,  
\item [\textit{Stab2:}] $S^\omega_\lambda(\lambda)=\prod_{(i,j)\in \op{inv}_w(\lambda)}(t_j-t_i)$, where
\begin{eqnarray*}
\op{inv}_w (\lambda)&=&\left\{(i,j)\mid 1\leq w^{-1}(i)<w^{-1}(j)\leq N, \lambda_i>\lambda_j\right\}.
\end{eqnarray*}
\item [\textit{Stab3}:]$\op{deg}(S^w_\lambda(\mu))=\ell(w^{-1}\la)$ for all $\mu\in\Lambda_n$ such that $\mu\geq_w\lambda$.
\end{itemize}
\end{lemma}
This class is called the {\it $w$-twisted Schubert class} (or {\it $w$-stable class}) {\it corresponding to $\lambda$}.
\begin{proof}
Note that $\ell(w^{-1}\la)$ equals the cardinality of the set
\begin{eqnarray*}
w(\op{inv}(w^{-1}(\la))&=&\{(i,j)\mid 1\leq w^{-1}(i)<w^{-1}(j)\leq N, \la_i>\la_j\}=\op{inv}_w (\la).
\end{eqnarray*}
Then the  proof is exactly as in \cite[Lemma 1]{KT03}, but with the order there replaced by $\geq_{w}$. 
\end{proof}
\begin{remark} 
\label{Delta}
For $w=e$ we obtain the ordinary equivariant Schubert classes. In particular, $\Delta_a S^e_\la=S^e_{s_a(\la)}$ if $\ell(\la)>\ell(s_a\la)$ and $\Delta_aS^e_\la=0$ otherwise, see \cite[Lemma 6]{KT03}.
\end{remark}
\begin{remark}
\label{important}
Alternatively one could define $S^w_\la(\mu)={}^w(S^e_{w^{-1}(\la)}(w^{-1}(\mu))$. 
Indeed, the characterizing properties \textit{Stab1} and \textit{Stab3} follows directly from the support condition for the untwisted Schubert classes. For the second one it suffices to show $w(\op{inv}_e \omega^{-1}(\lambda) )= \op{inv}_\omega \lambda.$ To see this note that $(i,j)\in w(\op{inv}_e \omega^{-1}(\lambda) )$ if and only if $(w^{-1}(i),w^{-1}(j))\in \op{inv}_e w^{-1}(\lambda)$. This is equivalent to $w^{-1}(i)<\omega^{-1}(j)$ with $w^{-1}(\lambda)_{w^{-1}(i)}>w^{-1}(\lambda_{w^{-1}(j)}$ or to $w^{-1}(i)<\omega^{-1}(j)$ with $\lambda_i>\lambda_j.$
But the latter means by definition that $(i,j)\in \op{inv}_w(\lambda)$.
\end{remark}

\begin{example}
Assume the setup of Example~\ref{exP1}. Then 
\begin{eqnarray*}
S^{e}_{01}(\mu)=
\begin{cases}
1,&\\
1,&
\end{cases}
S^{e}_{10}(\mu)=
\begin{cases}
(t_1-t_2),&\\
0,&
\end{cases}
S^{s_1}_{10}(\mu)=
\begin{cases}
1,&\\
1,&
\end{cases}
S^{s_1}_{01}(\mu)=
\begin{cases}
0,&\\
(t_2-t_1).&
\end{cases}
\end{eqnarray*}
where the first line is always for $\mu=(10)$ and the second for $\mu=(01)$.

\end{example}

\subsection{Wall-crossing}
We now express, for fixed $w\in S_N$, the $s_iw$-twisted equivariant Schubert basis in terms of the $w$-twisted one for any simple transposition $s_i$. This generalises and gives a geometric interpretation of the formula \cite[Proposition 3.2 with $l=1$]{MolevSagan} for factorial Schur functions.
\begin{proposition}
\label{startingpoint}
Let $w\in S_N$ and $s=s_i=(i,i+1)$ then 
\begin{eqnarray*}
S^{ws}_{\lambda} ({\mu})&=&
\begin{cases}
S^{w}_{wsw^{-1}(\lambda)}({\mu}) +(t_{w(i+1)}-t_{w(i)})S^{w}_{\lambda}({\mu})& \text{if }\la_{w(i+1)}>\la_{w(i)},\\
S^{w}_{wsw^{-1}(\lambda)}({\mu})&\text{otherwise.}\\
\end{cases}
\end{eqnarray*}
In particular, the $\{S^{w}_{\lambda}\}_{\lambda\in \Lambda_n}$ form, for fixed $w\in S_N$, a basis of $H^*_T(\Xn)$. 
\end{proposition}
\begin{proof}
We call the function on the right hand side $\alpha$ and rewrite $\alpha(\mu)$ using  Remark~\ref{important}. Abbreviating $\nu=w^{-1}(\la)$ we obtain
\begin{eqnarray*}
\alpha({\mu})&=&
\begin{cases}
w (S^{e}_{s\nu}(w^{-1}\mu) +(t_{i+1}-t_{i}) S^{e}_{\nu}(w^{-1}\mu))& \text{if }\nu_{i+1}>\nu_{i},\\
w (S^{e}_{s\nu}(w^{-1}\mu))&\text{otherwise.}
\end{cases}
\end{eqnarray*}
In the first case we have $\ell(s\nu)>\ell(\nu)$ hence using Remark~\ref{Delta} we obtain 
$\alpha({\mu})={}^w((s\cdot S^e_{s\nu})(w^{-1}\mu))={}^{ws}(S^e_{s\nu}(sw^{-1}\mu))=S^{ws}_\la(\mu)$ as claimed.
 In the second case there are two possibilities $\nu_{i+1}=\nu_{i}$, or $\nu_{i+1}<\nu_{i}$ which means in either case $\Delta_i(S^{e}_{s\nu})=0$ by Remark~\ref{Delta}. Hence $S^{e}_{s\nu}=s\cdot S^{e}_{s\nu}$ and therefore $w( S^{e}_{s\nu}(w^{-1}\mu))=ws(S^{e}_{s\nu}(sw^{-1}\mu))=S^{ws}_{\lambda} ({\mu})$. This finishes the proof.
\end{proof}

Identify  $H^*_T(\Xn)$ with $\mV_{N,n}$ from Section~\ref{secYBA} by sending $S^{e}_\lambda$ to $v_\lambda$ and extending $\bP$-linearly. Then Proposition~\ref{startingpoint} can be formulated in terms  of the $R$-matrix $R$ from \eqref{L1L2} as follows.
\begin{corollary}
\label{corstartingpoint}
The base change from $\{S^{ws_i}_{\lambda}\}_{\lambda\in \Lambda_n}$ to $\{S^{w}_{\lambda}\}_{\lambda\in \Lambda_n}$ is given by $R_{a,b}(t_b,t_a)$ acting on the $a$th and $b$th tensor factor of $\mV_{N}$, where $a=w(i)$ and $b=w(i+1)$.
\end{corollary}

\begin{remark}
Corollary~\ref{corstartingpoint} can be viewed as a construction of the $R$-matrix in the spirit of the geometric construction of the $R$-matrices of the Yangians in \cite{MO}.  In light of Remark~\ref{weyl}, Proposition~~\ref{corstartingpoint} can be viewed as the equivariant behaviour of attracting manifolds when we cross a wall passing from one Weyl chamber to a neighboured one. In the non-equivariant version this wall-crossing collapses to a (boring) symmetric group action (permuting the factors of $V^{\otimes N}$).
\end{remark}

By the equivariant localisation theorem, $H_T^*(\Xn)\subset \Func(n)$ becomes an isomorphism after localisation at $\op{loc}$ with  {\it fixed point-basis} (over $\bP^{\op{loc}}$)  $\mathbbm{1}_ p\in  \Func(n)$, $p\in \Xn^T$,   given by  $\Eins_ p(p')=\delta_{p,p'}$. 

\subsection{Geometric Bethe basis}
To connect the algebraic setup from the previous sections to the geometric one, we have to make (a common) twist, which means we consider instead of the ordinary Schubert basis the twisted Schubert basis $\mathbb{B}_n$ of $H^*_T( \op{X})$ corresponding to the longest permutation 
\begin{eqnarray*}
\omega_0&=&
\begin{pmatrix}
1&2&\ldots&N \\
N&n-1&\ldots&1
\end{pmatrix}
\end{eqnarray*}
By {\it Stab 1} to {\it Stab 3} the twisted Schubert class $S^{\omega_0}_{({0}^k,{1}^n)}\in\mathbb{B}_n$   takes the value 
\begin{eqnarray}
\label{beta}
\beta_n&=&\prod_{1\leq i \leq k,\, k+1\leq j\leq N}(t_j-t_i)
\end{eqnarray}
on the fixed point $p_n\in \Xn ^T$ corresponding to the unique maximal element $\zeta=({0}^k,{ 1}^n)=(\underbrace{0,\ldots,0}_\text{k},\underbrace{1,\ldots,1}_\text{n})\in \Lambda_n$ and is $0$ on all other fixed points. The action \eqref{symactioneq} of the symmetric group in this basis is easy to calculate: 
\begin{lemma}
\label{sym2}
For any $i\in[N-1]$ and $\la\in\Lambda_n$ we have 
\begin{eqnarray}
{s}_i\cdot(S^{\omega_0}_\lambda)&=&
\begin{cases}
 S^{\omega_0}_\lambda-(t_{i+1}-t_{i}) S^{\omega_0}_{{\bf s}_{i}\cdot\lambda},& \text{if } {\bf s}_{i}\lambda>\lambda,\\
\omega_0\cdot S_{\omega_0\lambda},& \text{otherwise.}
\end{cases}
\end{eqnarray}
\end{lemma}
\begin{proof}
Note that $s_i\omega_0=\omega_0s_{N-i}$ and $\omega^{-1}_0=\omega_0$ holds.
Thus by Remark~\ref{important} ${s}_i\cdot S^{\omega_0}_\lambda={s}_i\cdot \omega_0\cdot S_{\omega_0(\lambda)}=(s_i \omega_0)\cdot S_{\omega_0\lambda}
=(\omega_0s_{N-i+1})\cdot S_{\omega_0(\lambda)}=\omega_0\cdot s_{N-i+1}\cdot S_{\omega_0(\lambda)}$. On the other hand, expressing \eqref{symactioneq} via the last part of Lemma~\ref{Demazurelemma}, gives
\begin{eqnarray*}
{\bf s}_{N-i}\cdot S_{\omega_0\lambda}&=&
\begin{cases}
 S_{\omega_0\lambda}-((t_{N+1-i}-t_{N-i})S_{{\bf s}_{N-i}\cdot\omega_0\lambda})& \text{if }\,{ s}_{N-i}\omega_0(\lambda)<\omega_0(\lambda),\\
 S_{\omega_0\lambda}&\text{otherwise.}
\end{cases}
\end{eqnarray*}
Thus we obtain
\begin{eqnarray*}
\omega_0\cdot s_{N-i+1}\cdot S_{\omega_0(\lambda)}
&=&\begin{cases}
\omega_0\cdot S_{\omega_0(\lambda})-(t_{i+1}-t_{i}) \omega_0\cdot S_{\omega_0\cdot { s}_{i}(\lambda)}& \text{if }\,{s}_{N-i}\omega_0(\lambda)<\omega_0(\lambda),\\
\omega_0\cdot S_{\omega_0(\lambda})& \text{otherwise}\\
\end{cases}\\
&=&\begin{cases}
 S^{\omega_0}_\lambda-(t_{i+1}-t_{i}) S^{\omega_0}_{{s}_{i}(\lambda)}& \text{if }\,{ s}_{i}\lambda>\lambda,\\
\omega_0\cdot S_{\omega_0(\lambda)}& \text{otherwise}.\\
\end{cases}
\end{eqnarray*}
This proves the claim
\end{proof}
Comparing it with \eqref{sym} we obtain the following result: consider the localisations $\mV_N^{\op{loc}}$ from \eqref{loc} and $H_T^*(\Xn)^{\op{loc}}=\Func(n)^{\op{loc}}$ as $\bP^{\op{loc}}$-modules. 
 
\begin{corollary} 
\label{Bethe}
Sending the standard basis element $v_\la$ to the twisted Schubert basis element  $S^{\omega_0}_\lambda$ defines an isomorphism of 
$\bP^{\op{loc}}$-modules 
\begin{eqnarray}
\label{DefPsi}
\Psi:\quad\mV_N^{\op{loc}}&\cong& \bigoplus_{n=0}^N H_T^*(\Xn)^{\op{loc}}
\end{eqnarray} 
Under this isomorphism 
$\Psi$, the normalised Bethe vectors are sent to the fixed point basis vectors, in formulae $\Psi(\bb_\la)=\Eins_{\bw_\la}$, where $\bw_\la$ is the torus fixed point corresponding to $\la$.
\end{corollary}
\begin{proof}
The first claim follows directly from the definitions. By \eqref{sym} and Lemma~\ref{sym2} the isomorphism $\Psi$ intertwines the two $S_N$-actions. But then the statements follow from Lemma~\ref{hvector} as the starting point together with Proposition~\ref{SNonBethe}, since \eqref{symactioneq} becomes the obvious permutation action of $S_N$ on torus fixed points. 
\end{proof}

\section{The geometric Yang-Baxter algebra via convolution in cohomology}
Having established (via Corollary~\ref{corstartingpoint}) a homology interpretation of the $R$-matrix $R(x,t)$,  we will provide a similar interpretation of the corresponding Yang-Baxter algebra $\op{YB}_N$.

\subsection{Some combinatorics of torus fixed points in partial flag varieties}
In order to state the main theorem we will need not only the Grassmannians $\Xn$ but also their generalisation, the partial flag varieties. For any positive integer $m$ denote by $ \mathcal {P}_m(N)$ the set of compositions  of $N$ with $m$ parts (possibly zero), e.g.\  $\mathcal {P}_2(2)=\{(2,0),(1,1),(0,2)\}$. For $\bd \in \mathcal P_m(N)$ let $ \Xd$ be the corresponding {\it partial flag variety}, i.e.\   the projective variety given by all flags
\begin{eqnarray}
\label{flags}
\op{F}&=& (\{0\} \subset F_1 \subset \ldots \subset F_{m-1} \subset F_m = \mathbb C^N),
\end{eqnarray}
such that $\op{dim} (F_{i+1}/F_i) = d_i$. That is $ \Xd=G/P$ where $P$ is the parabolic subgroup given by all matrices $(a_{i,j})_{i,j}$ with $a_{i,j}=0$ if $j<\mu_r<i$ for some $r$, and $T$ still acts on it.  
There are $m$ natural $T$-equivariant vector bundles $M_1, \ldots,M_m$ over $\Xd$ of rank $d_1, \ldots, d_m$ respectively; the fibre of $M_i$ at  $\op{F}$ is the  vector space $F_{i+1}/F_i$. These bundles, as well as the tangent bundle $\op{T}(\Xd)$, are $T$-equivariant and thus have equivariant Chern classes which we will denote by $\mathfrak{c}_i(M)$, $\mathfrak{c}_i(\op{T}(\Xd))$.\\

We recall some well-known facts we will need later. The main reference for us is \cite{Anderson}.
The set of $T$-fixed points $\Xd^T \subset \Xd$ is finite. To describe them explicitly let $\op{F}^e_\bd$ be the standard flag where $F_i$ is spanned by $\{e_1,\ldots,e_{\mu_i} \}$, where  $\mu_i = \sum_{j=1}^i d_i$. Clearly,  $\op{F}^e_\bd$ is fixed under the action of $T$. The symmetric group $S_N$ acts on $\Bbb C^N$ by permuting the elements of the standard basis, so it acts on $\Xd$. For any $w \in S_N$, the partial flag $\op{F}^w_\bd= w(\op{F}^e_\bd)$ belongs to $\Xd^T$, and in this way we get all fixed points. The stabiliser of $\op{F}^e_\bd$ coincides with the Young subgroup $S_\bd=S_{d_1}\times S_{d_2}\times \ldots\times S_{d_m}$, so the assignment $w \mapsto \op{F}^w_\bd$ induces a bijection between
$S_N/S_\bd$ and $\Xd^T$. In other words, a fixed point can be described as a partitioning $\bw$ of the set $[N]$ into $m$ disjoint subsets $\bw_i$ of cardinalities $d_i$ such that $F_{i}/F_{i-1}$ is spanned by the images of the elements $e_l$, $l\in \bw_i$, see \cite{Fulton} for more details. 

\begin{definition} 
\label{notafixed}
As a special case, a fixed point for $\Xn$ is determined by $\bw_1$ (exactly the positions of  $1$'s in the corresponding $\{0,1\}$-word),  and then of $\op{X}_{n,1,{n-k-1}}$ by a tuple $\bw_z=(\bw_1,z)$ for some $z\in [N], l\notin\bw_1$, and in $\op{X}_{n+1}$ then by the union $\bw_1\cup\{z\}$ which we abbreviate as $\bw_1\cup z$. For $S\subset [N]$ let $\ov{S}$ be the complement in $[N]$.  
\end{definition}
For instance if $N=4$, then $X_2$ has six fixed points $\langle e_i, e_j\rangle$, $1\leq i<j\leq 4$ corresponding to the partitioning $\bw=\{i,j\}\cup \overline{\{i,j\}}$ and determined by $\bw_1=\{i,j\}$. A fixed point in $X_{2,1,1}$ is then given by $\langle e_i, e_j\rangle\subset\langle e_i, e_j,e_z\rangle$ for some distinct $i,j,k\in\{1,2,3,4\}$ and so determined by the tuple $(\{i,j\},z)$.

 \subsection{The geometric Yang-Baxter algebras}
For $\bw\in\Xd^T$ let $\inc_{\bw} : \{\op{pt}\} \hookrightarrow \Xd$, $\op{pt}\mapsto\bw$ be the $T$-equivariant inclusion. We have the following standard fact, see e.g. \cite[Example 2.1]{Anderson}.

\begin{lemma}\label{fixedpts}
Let $\bw=(\bw_1,\ldots,\bw_m)\in\Xd^T$. Then the  Euler class $\eu_{\bw}$ of the fixed point $\bw$ is 
\begin{eqnarray*}
\eu_{\bw}:=\inc_{\bw}^{*}(\mathfrak{c}_{\op{top}}(\op{T}(X_\bd)))& =& \prod\limits_{m\geq i>j\geq 1} \prod\limits_{a\in \bw_i,\, b\in \bw_j} (t_a - t_b),
\end{eqnarray*}
where $\op{top}$ is twice the dimension of $X_{\bd}$ and $\inc_{\bw}^{*}(\mathfrak{c}_j(M_i))= \sigma_j(t_{j_1},\ldots,t_{j_{d_i}}),$ where $\bw_i=\{j_1,\ldots ,j_{d_i}\}$ and $\sigma_j$ denotes the $j$th elementary symmetric function.
\end{lemma}

For $m=1$, the partial flag variety is a point and for $m\leq 2$ we obtain the Grassmannian varieties $\Xn$. The inclusion $\incl: \op{X}^T\hookrightarrow\op{X}$ induces a monomorphism $\incl^*$ on equivariant cohomology which is exactly $\op{res}$. Moreover $M_1$ is the tautological bundle $\mathcal T_n$, and $M_2$ the quotient bundle $\mathcal Q_n$.
 For the three step flag variety $\op{X}_{(n,1,N-n-1)}$ we have a diagram
 \begin{eqnarray}
\label{correspondence}
\xymatrix{
&  \op{X}_{(n,1,N-n-1)}\ar[dl]_{\pi_1}\ar[dr]^{{\pi_2}}\\
\Xn && \op{X}_{n+1}
}
\end{eqnarray}
 
where $\pi_1$ and  ${\pi_2}$ are the obvious proper projection maps. It defines convolution operators acting on $H^*_T(\op{X})$, where $\op{X}=\bigcup_{n=0}^N \Xn$ is the disjoint union of all Grassmannians. The {\it convolution algebra} $\mathcal C$  is the subalgebra of endomorphisms of $H^*_T(\op{X})$ generated by the endomorphisms $ {\pi_2}_{*}( \alpha\cdot)\pi_1^*$ and ${\pi_1}_*(\alpha\cdot) {\pi_2^*}$ for $\alpha \in H^*_T( \op{X}_{(n,1,N-n-1)})$. In particular, $\mathcal{C}$ contains the following crucial elements,
\begin{eqnarray}
\label{bc}
b_n={\pi_2}_{*}{\pi_1^*}\quad\text{and}\quad c_{n}={\pi_1}_{*} {\pi_2^*}.
\end{eqnarray} 

\begin{definition}
\label{ABCDgeom}
Introduce the generating functions, for $0\leq n\leq N$,
\begin{eqnarray*}
\begin{array}[t]{lclclcl}
\mathcal{A}_n(x)&=&\sum\limits_{i=0}^{n}\mathfrak{c}_i(\mathcal{T}_n)x^{n-i}&&
\mathcal{A}_n^{\prime}(x)&=&\sum\limits_{i=0}^{k} (-1)^i\mathfrak{c}_i(\mathcal Q_n)x^{k-i},\\
\mathcal{B}_n(x)& =& b_{n}\mathcal{A}_{n}(x),&&\mathcal{B}^{\prime}_n(x)&=&\mathcal{A}^{\prime}_{n+1}(x)b_n,\\
\mathcal{C}_n(x) &=&  \mathcal{A}_{n}(x)c_n,&&\mathcal{C}^{\prime}_n(x)&=&c_n \mathcal{A}^{\prime}_{n+1}(x),\\
\mathcal{D}_{n+1}(x)&=&b_{n}\mathcal{A}_{n}(x)c_n,&&\mathcal{D}^{\prime}_{n}(x)&=&c_{n}\mathcal{A^{\prime}}_{n+1}(x)b_n,
\end{array}
\end{eqnarray*}
with the convention that for $n+1>N$ the corresponding operators are $0$.
\end{definition}

\begin{definition} The {\it geometric Yang-Baxter algebras} $\mathcal{YB}_N$  and $\mathcal{YB}'_N$ are defined to be the subalgebras of $\mathcal{C}$ generated by the coefficients of $\mathcal{A}_n(x),\mathcal{B}_n(x), \mathcal{C}_n(x),\mathcal{D}_n(x)$, 
respectively  of $\mathcal{A}^{\prime}_n(x),\mathcal{B}^{\prime}_n(x),\mathcal{C}^{\prime}_n(x),\mathcal{D}^{\prime}_n(x)$, where $0\leq n\leq N$.
\end{definition}
\begin{theorem} (Geometric Yang-Baxter algebra) \label{main} 
The algebras  $\mathcal{YB}_N$  and $\mathcal{YB}'_N$ are isomorphic to the Yang-Baxter algebras  $\op{YB}_N$  and  $\op{YB}_N'$ respectively from Definition~\ref{DefYBA}.
\end{theorem}
We will prove the theorem for the algebra $\mathcal{YB}_N$, the algebra $\mathcal{YB}'_N$ can be treated similarly. We first check the last two relations listed in Proposition~\ref{PropYBA} for the geometric operators $\mathcal A(x), \mathcal B(x), \mathcal C(x), \mathcal D(x)$ and then identify them with the algebraic ones from Section~\ref{secYBA}. Finally we show hat this defines the operator $\mathcal D(x)$ uniquely, matching it with the algebraic operator $D(x)$ in Lemma~\ref{Doperator}. This will finish the proof at the end of the section. 

\subsection{Some relations}
Checking the last two identities from (\ref{PropYBA}) will be a standard fixed point calculation which goes back to the work of Atiyah and Bott \cite[(3.8)]{AB}. In particular we need the Atiyah-Bott-Berline-Vergne integration formula. Let $X$ and $Y$ be compact non-singular algebraic $T$-varieties and $f :X \rightarrow Y$ be a $T$-equivariant morphism, then we have for all $\alpha \in H^{*}_T(X)$, (using notation from Definition~\ref{notafixed} and the restriction $\op{res}$ to fixed points, 
\begin{eqnarray}
\label{AB}
f_*(\alpha) &=&\op{res}^{-1}\left(\sum_{y\in Y^T} \sum_{\{x\in X^T\,|f(x)=y\}}\frac{\eu_y}{\eu_x}\,\,\inc_x^{*}\,a\right).
\end{eqnarray}

We apply this to the maps ${\pi_1}$ and ${\pi_2}$ from \eqref{correspondence}.  The involved ratios of the Euler classes are
\begin{eqnarray*}
\frac{\eu_{{\bw}}}{\eu_{{\bw}_z}}=\prod_{j\in \ov{{\bw}_1\cup z}}(t_z-t_j)^{-1}
&\text{and}&
\frac{\eu_{{\bw}_1\cup z}}{\eu_{{\bw}_z}}=\prod_{i\in  {\bw}_1}(t_i-t_z)^{-1}.
\end{eqnarray*}
The following explicitly computes the operators from Definition~\ref{ABCDgeom}.
\begin{proposition}\label{operators} In the fixed point basis the following holds for $\bw\in\Xn^T$,  $z\in \ov{\bw}$, and $z'\in\bw$. 
\begin{enumerate}[1.)]
\item The operator $\mathcal A(x)$ is represented by a diagonal matrix with the eigenvalues
\begin{eqnarray}
\label{Aeigen}
\mathcal A_n(x)(\Eins_{\bw})&=&\prod_{j\in {\bw}}(x+t_j)\Eins_{\bw}.
\end{eqnarray}
\item The $({\bw}_1\cup z,\bw_1)$-matrix entry of $b_n$ is $\prod_{j\in {\bw_1}}(t_{j} - t_z)^{-1}$, the $({\bw}_1,{\bw}_1\setminus z')$-matrix entry of $c_n$ equals $\prod_{j\in \ov{{\bw}_1}}(t_{z'}-t_j)^{-1}$. All other matrix entries are zero.
\end{enumerate}
\end{proposition}
\begin{proof}
By Lemma~\ref{fixedpts}, the restriction of $c_i(\mathcal T_n)$ (respectively  $c_i(\mathcal Q_n)$) to ${\bw}$ is equal to the $i$th elementary symmetric function $\sigma_i$ in the variables $t_j$, $j\in {\bw}_1$ (respectively $j\in \ov{\bw_1}$). Thus \eqref{Aeigen} holds. If $\bw_1$ denotes a fixed point in $\Xn^T$, then ${\pi_1}^{-1}(\bw_1)\cap \op{X}_{(n,1,N-n-1)}=\{ \bw_1\cup z\mid z\in\ov{\bw}_1\}$ and thus the asserted matrix coefficients ${\pi_2}_*{\pi_1^*}$ follow from \eqref{AB}, and similarly for  ${\pi_1}_*{\pi_2^*}.$ \end{proof}

\begin{proposition}
\label{proprel}
The following identities hold
\begin{enumerate}[1.)]
\item \label{idone} $\mathcal C(x_2)\mathcal B(x_1)-\mathcal C(x_1)\mathcal B(x_2) =(x_1-x_2)(\mathcal A(x_1)\mathcal D(x_2)-\mathcal D(x_2)\mathcal A(x_1)),$
\item \label{idtwo}
$\mathcal D(x_1)\mathcal C(x_2) =\mathcal D(x_2)\mathcal C(x_1),$
\item \label{idthree}
$\mathcal C(x_1)\mathcal D(x_2)-\mathcal C(x_2)\mathcal D(x_1) =(x_1-x_2)\mathcal D(x_2)\mathcal C(x_1)$.
\end{enumerate}
\end{proposition}
\begin{proof} We start with \ref{idone}. The map $c_n b_n$ sends $\Eins_{\bw}$ to a linear combination of basis vectors, first via $b_n$ such vectors attached to fixed points of the form  $\bw_1\cup z$ for some $z\in \ov{\bw_1}$  and then via $c_n$ to those of the form 
$(\bw_1\cup z)\setminus z'$, $z'\in ({\bw_1}\cup z)$. In contrast,  $b_n c_n$ first maps to those of the form $\bw_1\setminus z''$ and then to  
$(\bw_1\setminus z'')\cup z'''$, where $z''\in \bw_1$ and $z'''\in\ov{(\bw_1\setminus z'')}$. 
To compare, pick the same vector on each side. Assume first that $z=z'''\not=z'=z''$. Then the denominator occurring in $c_n b_n$ is
equal to $\prod_{j\in {\bw}_1}(t_{j} - t_z)\prod_{j\in \ov{{\bw}_1\cup z}}(t_{z'} - t_j)$, 
whereas for $b_n c_n$ it equals $\prod_{j\in {\ov{{\bw}_1}}}(t_{z'} - t_j)\prod_{j\in {\bw}_1, j\not=z'}(t_{j} - t_z)$. They obviously agree. In case $z=z'''=z'=z''$ we obtain $\prod_{j\in{\bw}_1}(t_{j} - t_z)\prod_{j\in\ov{{\bw}_1}, j\not=z}(t_{z} - t_j)$ as common denominator on the left. In both cases the same common denominator appears also on the right hand side of  \ref{idone}. Now we compare the numerator  of the coefficient of $\Eins_{(\bw_1\cup z)\setminus z'}$ appearing on the left with the numerator in the  coefficient of  $\Eins_{(\bw_1\setminus z'')\cup z''}$ appearing on the right (assuming the basis vectors agree). The first equals
\begin{eqnarray}
&&\prod_{j\in  ({\bw_1}\cup z)\setminus z'}\prod_{i\in {\bw}_1}(x_2+t_j)(x_{1}+t_i) - \prod_{j\in  ({\bw}_1\cup z)\setminus z'}\prod_{i\in {\bw}_1}(x_{1}+t_j)(x_2+t_i)\nonumber\\
&=&(x_1-x_2)(t_z-t_{z'})\prod_{i,\,j\in  ({\bw_1}\setminus z')}(x_{1}+t_i)(x_2+t_j)\label{boring}
\end{eqnarray}
whereas the latter is equal to
\begin{eqnarray*}
&&(x_1-x_2)\left(\prod_{j\in  ({\bw}_1\setminus z')\cup z}\prod_{i\in  ({\bw}_1\setminus z')}(x_1+t_j)(x_2+t_i) - \prod_{j\in  ({\bw}_1\setminus z')}\prod_{i\in  {\bw}_1}(x_2+t_j)(x_1+t_i)\right).
\end{eqnarray*}
which coincides with \eqref{boring}. This proves  \ref{idone}. 

The operators from the identity  \ref{idtwo} act on our basis elements by removing two elements from  ${\bw}_1$ and then add one element to the resulting set. More precisely, $\mathcal{C}$ gives a linear combination of basis vectors corresponding to  ${\bw}_1\setminus r$, $r\in {\bw}_1$, and $\mathcal{D}$ then of the form  $y=({\bw}_1\setminus \{r,s\})\cup u$ with $s\in ({\bw}_1\setminus r)$ and $u\in {\bw}_1\setminus \{r,s\}$. Let us calculate the coefficient of $\Eins_y$ in $\mathcal D(x_1)\mathcal C(x_2)\Eins_{\bw}$. There are two ways to get from $\bw_1$ to $y$ depending if we first remove $r$ or first $s$. By Lemma~\ref{operators} the denominators corresponding to the two summands differ precisely by a sign (coming from the factor $(t_r-t_s)$). Calculating, with Lemma~\ref{operators}, the numerator for the left hand side we get
\begin{eqnarray}
&&\prod_{j\in {\bw}_1\setminus \{r,s\}}\;\prod_{i\in {\bw}_1\setminus \{r\}}(x_1+t_j)(x_2+t_i)\; \;-\;\prod_{j\in ({\bw}_1\setminus \{s,r\})}\;\prod_{i\in ({\bw}_1\setminus \{s\})}(x_1+t_j)(x_2+t_i)\nonumber\\
&=&\label{secondi}
(t_s-t_r)\prod_{i,\,j\in ({\bw}_1\setminus \{r, s\})}(x_1+t_i)(x_2+t_j)
\end{eqnarray}
which is equal to numerator for the right hand side, 
\begin{eqnarray*}
\prod_{j\in {\bw}_1\setminus \{r,s\}}\;\prod_{i\in {\bw}_1\setminus \{r\}}(x_2+t_j)(x_1+t_i)\;\;-\;\prod_{j\in {\bw}_1\setminus \{s,r\}}\;\prod_{i\in {\bw}_1\setminus \{s\}}(x_2+t_j)(x_1+t_i).
\end{eqnarray*}

This proves  \ref{idtwo}, and we pass to \ref{idthree}.  Above we have calculated the matrix coefficients of the composition $\mathcal D(x_1)\mathcal C(x_2)$ for the pair $\bw$ and $y$. Now we calculate this coefficient for the operator $\mathcal C(x_1)\mathcal D(x_2)-\mathcal C(x_2)\mathcal D(x_1).$
These operators remove an element from ${\bw}_1$, add an element to the result and  remove one again.

As before observe that in $\mathcal C(x_1)\mathcal D(x_2)(\Eins_{\bw})$ there are two summands contributing to the coefficient of $\Eins_y$ whose denominators differ by a sign. The numerators of these summands contribute
\begin{eqnarray}
\label{Tea1}
\prod_{j\in {\bw}_1\setminus \{r ,s\}\cup u}\;\prod_{i\in {\bw}_1\setminus \{r\}}(x_1+t_j)(x_2+t_i)\;\;-\;\prod_{j\in {\bw}_1\setminus \{r ,s\}\cup u}\;\prod_{i\in {\bw}_1\setminus \{s\}}(x_1+t_j)(x_2+t_i).
\end{eqnarray}
Likewise for the operator ${\mathcal C(x_2)}{\mathcal D(x_1)}$  the numerators contribute
\begin{eqnarray}
\label{Tea2}
\prod_{j\in {\bw}_1\setminus \{r,s\}\cup u}\;\prod_{i\in {\bw}_1\setminus \{r\}}(x_2+t_j)(x_1+t_i) \;\;-\; \prod_{j\in {\bw}_1\setminus \{r,s\}\cup u}\;\prod_{i\in {\bw}_1\setminus \{s\}}(x_2+t_j)(x_1+t_i).
\end{eqnarray}
Taking the difference \eqref{Tea1}-\eqref{Tea2} we get
\begin{eqnarray*}
&((x_1+t_u)(x_2+t_s-x_2-t_r)-(x_2+t_u)(x_1+t_s)+(x_2+t_u)(x_1+t_r))
\prod_{i,\,j\in {\bw}_1\setminus \{r, s\}}(x_1+t_i)(x_2+t_j)&\\
&=\;(x_1-x_2)(t_s-t_r)\prod_{i,\,j\in {\bw}_1\setminus \{r, s\}}(x_1+t_i)(x_2+t_j).&
\end{eqnarray*}
Since this equals $(x_1-x_2)$ times \eqref{secondi} this finishes the proof.
\end{proof}

\subsection{The geometric action matches the algebraic action}
 Recall from Corollary~\ref{Bethe} the identification $\Psi: \mV_N^{\op{loc}}\cong H_T^*(\Xn)^{\op{loc}}$. Then the following holds.
 \begin{theorem}
 \label{match}
The action of the geometric Yang-Baxter algebra operators from Definition~\ref{ABCDgeom} in the standard basis of  $\mV_N^{\op{loc}}$ is given by the operators defined by the monodromy matrix $M_N$ from \eqref{Mono} for the pair $(R,L).$
\end{theorem}
Theorem~\ref{main} follows (with an isomorphism identifying the generators in the obvious way).

\begin{remark}
\label{degeneration}
It was pointed out in \cite{NSh} and \cite{BMO} that our Yang-Baxter algebras are certain limits of the Yangian for $\mathfrak{gl}_2$. Such limits can be calculated rigorously after an appropriate Drinfeld twist of the Yangian, in the sense of \cite[Section 1]{GKassel}. We do however not know a rigorous direct geometric construction of these twists.
\end{remark}

\begin{proof}[Proof of Theorem~\ref{match}]
We start with the operator $\mathcal A(x)$. It obviously commutes with the action \eqref{symactioneq} and $\Psi$ is $S_N$-equivariant.  Therefore, it is enough to calculate how  $\mathcal A(x)$ acts on $S^{\omega_0}_{\zeta}$ with $\zeta=({0}^k,{1}^n)\in\Lambda_n\hat{=}\Xn^T$  for any $0\leq n\leq N$.  Since, using \eqref{beta}, 
$S^{\omega_0}_{\zeta}$ corresponds to $\beta_n \Eins_\zeta$ in the fixed point basis, it is clear that the value of $\mathcal A(x)$ at this vector is the same value as the obtained using the diagrammatics of the first section. 

To prove the claim for the operator $\mathcal B(x)$ we will use the following obvious facts:
\begin{itemize}
\item $\mathcal B(x)S^{\omega_0}_\zeta$ is a linear combination of fixed point basis vectors $\Eins_{\mu(j)}$ with $\mu(j)$ given by the $\{0,1\}$-word $(\underbrace{0,\ldots,0}_\text{j-1},1,0,\ldots,0,\underbrace {1,\ldots,1}_\text{n})$.  
\item The expansion of $S^{\omega_0}_{\mu(j)}$ in the fixed point basis contains $\Eins_{\mu(1)}$ precisely for $j=1$.
\end{itemize}
Therefore all we need to calculate is the coefficient of $\Eins_{\mu(1)}$ in the expansion of $S^{\omega_0}_{{\mu(1)}}$ in order to calculate the appropriate matrix coefficient of $\mathcal B(x)$ in our $\omega_0$-twisted Schubert basis. This is readily done by {\it Stab2}, namely
$S^{\omega_0}_{\mu(1)}=\prod_{2\leq i\leq k, k+1\leq l\leq N}(t_l-t_i)\Eins_{\mu(1)}+\op{rest}  \op{terms}$.
Hence we obtain
\begin{eqnarray*}
\mathcal B(x)S^{\omega_0}_{\mu(1)}=\mathcal B(x)\beta_n\Eins_{\mu(1)}=\beta_n\frac{\prod_{k+1\leq l\leq N}(x+t_l)}{\prod_{k+1\leq l\leq N}(t_l-t_1)}\Eins_{\mu(1)}+\ldots=\prod_{k+1\leq l\leq N}(x+t_l)S^{\omega_0}_{\mu(1)}+\ldots.
\end{eqnarray*}
where the $\ldots$ indicates irrelevant terms. This matrix coefficient is exactly the same as the one calculated in the section \ref{secYBA} using the diagrammatics for the operator $B(x)$ defined by the monodromy matrix built out of the operator $L(x,t)$. The arguments for $\mathcal C(x)$ are analogous and therefore omitted. It remains to consider $\mathcal D(x)$. By Lemma~\ref{Doperator} below we only need to calculate the restrictions of the operator $D(x)$ and  the geometric operator $\mathcal D(x)$ to $H^*_T(\op{X}_N)$. It is an easy calculation. One obtains, using the diagrammatics, that $D(x)$ restricts to the identity operator. Let us calculate the restriction of the geometric operator. We have 
$\mathcal D(x)S^{\omega_0}_{(1,\ldots ,1)} =\sum_{r=1}^N\prod_{i\not=r}
\frac{x+t_{i}}{t_{i}-t_{r}}~S^{\omega_0}_{(1,\ldots  ,1)} =S^{\omega_0}_{(1,\ldots, 1)}
$
because the polynomial in $x$ in the numerator has degree $N-1$ and is equal to $1$ at the points $-t_1,\ldots,-t_N$, it cancels with the denominator. The theorem follows.
\end{proof}

\begin{lemma} \label{Doperator} The operator $\mathcal D(x)$ is uniquely determined by the operators $\mathcal A(x)$, $\mathcal B(x)$ and $\mathcal C(x)$ and its restriction to $H_T^*(\op{X}_N)$ using the relations from Proposition~\ref{proprel}.
\end{lemma}

\begin{proof} In the fixed point basis, $\mathcal A(x)$ is diagonal with pairwise distinct eigenvalues. Therefore,
\begin{eqnarray*}
\mathcal B(x_1)\mathcal C(x_2)-\mathcal B(x_2)\mathcal C(x_1) &=&(x_1-x_2)(\mathcal D(x_1)\mathcal A(x_2)-\mathcal A(x_2)\mathcal D(x_1))
\end{eqnarray*}
defines the off-diagonal entries of $\mathcal D(x)$. To define the diagonal entries we use
\begin{eqnarray*}
\mathcal C(x_1)\mathcal D(x_2)&-&\mathcal C(x_2)\mathcal D(x_1)=(x_1-x_2)\mathcal D(x_2)\mathcal C(x_1)
\end{eqnarray*}
It allows us from the known off-diagonal entries to determine uniquely the diagonal entries of $\mathcal D(x)$, because the matrix of $\mathcal C(x)$ does not contain a zero row in the fixed point basis, and thus for any index $i$ we find an index $j$ such that the $(i,j)$-entry in $C(x)$ is non-zero and and allows to recover the $i$th diagonal entry of  $\mathcal D(x)$.
\end{proof}

\subsection{Explicit formulae}
\label{sec:explicitformulas}
Thanks to Theorem~\ref{match} we do not need to distinguish anymore (also with respect to notation)  between the algebraic and geometric operators and can freely pass between the two descriptions when we do our calculations. Using the graphical calculus one can easily calculate the  action of $\operatorname{YB}_N$ and $\operatorname{YB}'_N$ on an extremal weight vector. Then the $S_N$-action allows us to directly obtain the following explicit formulae for any standard basis vector.
 
\begin{corollary} \label{explicitformulas} Let $0\leq n\leq N$ and $\zeta=(0^{N-n},1^n)$. With the identification \eqref{DefPsi} we have
\begin{align*}
&\mathcal A(x)v_{\zeta} =\prod_{j=k+1}^N(x+t_j)\;v_\zeta,&&
\mathcal A'(x)v_{\zeta} = \prod_{j=1}^k(x-t_j)\;v_\zeta,\\
&\mathcal B(x)v_{\zeta} =\prod_{j=k+1}^N(x+t_j)\;v_{10\ldots 01\ldots 1},&& 
\mathcal B'(x)v_{\zeta} =\sum_{i=1}^k\prod_{j=i+1}^k(x-t_j)\;v_{0\ldots 0\underset{i}{1}0\ldots01\ldots 1},\\
&\mathcal C(x)v_{\zeta} =\sum_{i=k+1}^N\prod_{j=k+1}^{i-1}(x+t_j)\;v_{0\ldots01\ldots1\underset{i}{0}1\ldots 1},&&
\mathcal C'(x)v_{\zeta} =\prod_{j=1}^k(x-t_j)\;v_{0\ldots 01\ldots 10},\\
&\mathcal D(x)v_{\zeta} =\sum_{i=k+1}^N\prod_{j=k+1}^{i-1}(x+t_j)\;v_{10\ldots01\ldots1\underset{i}{0}1\ldots 1},&&
\mathcal D'(x)v_{\zeta} =\sum_{i=1}^k\prod_{j=1}^k(x-t_j)\;v_{0\ldots 0\underset{i}{1}0\ldots 01\ldots 10}.
\end{align*}
\end{corollary}

 By Definition~\ref{ABCDgeom} we have the relations $\mathcal D_n(x)=b_n(\mathcal A_n(x))c_n$ and $\mathcal D'_n(x)=c_n(\mathcal A'_n(x))b_n$.  Corollary~\ref{Bethe}, formula \eqref{beta} and again the $S_N$-action allows us then to compute both Yang-Baxter algebra actions completely and explicitly in the Bethe basis, by computing it on an extremal weight vector (for instance via a $T$-fixed points calculation using Corollary~\ref{Bethe}). 

\begin{corollary} With the assumptions from Corollary~\ref{explicitformulas} we have
\label{cor:ABCDv0}
\begin{eqnarray*}
A(x)\bb_{\zeta}
&=&\prod_{i=k+1}^{N}(x+t_{i})~\bb_{0\ldots 01\ldots 1},
\\
 B(x)\bb_{\zeta} &=&\sum_{r=1}^k\prod_{i=k+1}^{N}%
\frac{x+t_{i}}{t_{i}-t_{r}}~\bb_{0\ldots 0\underset{r}{1}0\ldots 01\ldots  1}, \\
 C(x)\bb_{\zeta} &=&\sum_{r=k+1}^{N}\left(
\frac{1}{x+t_r}\prod_{i=k+1}^{N}(x+t_{i})\prod_{i=1}^{k}\frac{1}{t_{r}-t_{i}}\right)
~\bb_{0\ldots 01\ldots 1\underset{r}{0}1\ldots 1}.
\end{eqnarray*}
\end{corollary}

\begin{remark}
There also exists an explicit formula for the ${D}$-operator, however, it is slightly more involved. Employing the same argument as for $A$, $B$ and $C$, one finds
\[
D(x)\bb_{\zeta} =d_\zeta(x) \bb_{\zeta}+\sum_{a=1}^k\sum_{b=k+1}^N\frac{\prod_{k+1\leq j\leq N, j\not=b}(x+t_j)}{\prod_{j=k+1}^N(t_j-t_b)}\prod_{i=1}^k\frac{1}{t_a-t_i} ~\bb_{0\ldots 0\underset{a}1 0\ldots 01\underset{b}{0}1\ldots 1},
\]
where the coefficient $d_\zeta(x)$ cannot be determined by the symmetric group action. Instead one uses the graphical calculus from Section~\ref{diagrams} for the operator $D(x)$ and Theorem \ref{main} to find the missing matrix element and obtains, with $\lambda(i)=10\ldots 01\ldots 1\underset{i}{0}1\ldots 1$, the formula
\[
d_\zeta(x)=\sum_{i=k+1}^{N} S^{\omega_0}_{\lambda(i)}(\zeta)\prod_{j=k+1}^{i-1}(x+t_j).
\]

\end{remark}
In the same way we also obtain the following formulae.

\begin{corollary} With the same assumptions as in Corollary~\ref{cor:ABCDv0} we have
\begin{eqnarray*}
A'(x)\bb_{\zeta}
&=&\prod_{i=1}^{k}(x-t_{i})~\bb_{0\ldots 01\ldots 1},\\
B'(x)\bb_{\zeta} &=&\sum_{i=1}^k\frac{\prod_{1\leq j\leq k, j\not=i}(x-t_j)}{\prod_{j=k+1}^N(t_j-t_i)}~\bb_{0\ldots 0\underset{i}1 0\ldots 0 1\ldots 1},\\
 C'(x)\bb_{\zeta} &=&\sum_{j=k+1}^N\prod_{i=1}^k\frac{x-t_i}{t_j-t_i}~\bb_{0\ldots 01\ldots1 \underset{j}0 1\ldots 1}.
\end{eqnarray*}
\end{corollary}

\section{Equivariant quantum cohomology}
\label{secquantum}
As an application we now connect our approach with equivariant quantum cohomology of Grassmannians, see e.g. \cite{FP} for an introduction to quantum cohomology.  Equivariant quantum Schubert calculus was originally introduced by Givental and Kim, \cite{GivKim}, \cite{G}, \cite{K}. Here we focus on the special case of Grassmannians and will follow mainly the setup and conventions in \cite{Mich1}. In our setting of a quantum integrable system the transition to equivariant quantum cohomology corresponds to considering the operator $A(x) + q D(x)$ rather than just $A(x)$ or $D(x)$ by itself, where $q$ is the quantum (deformation) parameter. This approach is a special case of the construction in \cite{GoKo} but the new result here is the geometric construction of the YB algebra.

\subsection{The operators $D$}
Recall the expansion $D(x)=\sum_{i\geq 0}D^{(i)}x^{i}$ from \eqref{ABCD}. The summands $D^{(i)}$ all preserve the weight spaces of $\mathbb{V}_N$. We are interested in the largest possible exponent for $x$ where the action is not trivial. 

\begin{proposition}\label{D} 
The action of $D^{(n-1)}=D^{(n-1)}_N$ restricted to $\mathbb{V}_{N,n}$ is given in the standard basis $v_{\bf \la}$, $|{\bf \la}|=n$ from \eqref{basisvec} by 
\begin{eqnarray*}
D^{(n-1)} v_{(\la_1\la_2\la_3\cdots\la_{N-1}\la_N})
&=&
\begin{cases}
\;v_{(\la_N\la_2\la_3\cdots\la_{N-1}\la_1)}& \text{if } \la_{N}=1 {\text{ and }}n\geq 1,\\
\;0&\text{otherwise}.
\end{cases}
\end{eqnarray*}
\end{proposition}
In other words the operator $D_{N,n}^{}$ acts on basis vectors $v_{\la}$ by the usual action of the affine generator in the affine symmetric group $\widehat{S_N}$. 

\begin{remark}
With the standard identification of $\la\in\{0,1\}^N$, $|{\la}|=n$ with Young diagrams in a $(n\times (N-n))$-box,  this operator removes the rim hook of the length $N-1$ if it exists and is $0$ otherwise. The rim-hook algorithm is originally due to \cite{BCFF} in the non-equivariant case. The discussion with the D-operator can be found in \cite[Lemmas 3.12, 3.14 and 3.20]{Korff}. The equivariant case of the rim-hook algorithm has been discussed more recently in \cite{BBT}. The formulation using (algebraic) $D$-operators appears already in \cite{GoKo}, which also covers the case of equivariant quantum K-theory.
\end{remark}
\begin{proof}
The claim follows easily from the diagrammatic calculus described in Section~\ref{secYBA}. Indeed let us assume that the basis vector $v_{\mu}$ appears in the expansion of $D^{(n-1)}v_{\la}$. Since the case $n=0$ is obvious, let us assume $n\geq 1$. Then looking at the non-zero weights \eqref{sixteen} of our model we conclude that $\la_N=1=\mu_1$.  In case $N=n=1$ we must therefore have a single red cross and we are done. So let $N>1$. By our assumption on the degree of $x$ we need $n-1$ crossings made out of a vertical red line and a horizontal black line which implies already $\la_i=\mu_i$ for $1<i<N$ and  $\la_i=1$ and moreover $\la_1=0$. This means that the first cross has a black horizontal output edge and therefore from there onwards $\la_i=0$ implies $\mu_i=0$ with the corresponding crossing being always completely black. Hence only the asserted vector  $v_{\mu'}$ can appear. Our argument also shows that it in fact appears with coefficient $1$.
This finishes the proof.
\end{proof}


\subsection{Quantum deformation}
To deform the multiplication on $H^*_T(\Xn)$ we first note that the  ${\bf P}$-algebra structure of  $H^*_T(\Xn)$ is completely determined by the pairwise products of the equivariant Chern classes of the tautological bundle which are encoded in $A(x)=\mathcal{A}(x)$ as a generating function. We deform this generating function with a parameter $q$ by considering $T(x)=A(x)+qD(x)$ acting now on the space $\mV_N[q]$ which we identify with  the $\mC[q]$-module $H^*_T(\op{X})[q]$ such that specialising $q=0$ provides the setup studied so far. 

We define an (associative) algebra structure $*$ on $H^*_T(\Xn)[q]$ by taking the coefficients of the expansion of $A(x)+qD(x)$ into powers of $x$ as generators. The following proposition shows that in this way we obtain a commutative ${\bf P}[q]$-algebra structure on $H_T^*(\Xn)[q]$ which specialises by construction to $H^*_T(\op{X})$ if we set $q=0$. Analogously we define the deformation $T'(x)=A'(x)+qD'(x)$ of $A'(x)$ and the corresponding algebra structure $*'$.  

\begin{proposition} The products $*$ and $*'$ are commutative, i.e. for $i,j\in\mathbb{Z}_{\geq 0}$ it holds
\begin{eqnarray*}
T^{(i)}T^{(j)}=T^{(j)}T^{(i)}&\text{ and}&{T'}^{(i)}{T'}^{(j)}={T'}^{(j)}{T'}^{(i)}. 
\end{eqnarray*}
\end{proposition}
 \begin{proof} Consider first $T(x)=A(x)+qD(x)$. Then the claim is equivalent to the commutation relations $T(x_1)T(x_2)=T(x_2)T(x_1)$ of generating series. The idea of the proof of the latter relations is standard in the literature on integrable systems; see \cite{Bax}. We recall it here for completeness. 
 
First note that the $R$-matrix $R(x,y)$ from Definition \ref{ex} satisfies the commutation relation
\[
\left(\begin{smallmatrix}
1&0\\
0 & q
\end{smallmatrix}\right) \otimes
\left(\begin{smallmatrix}
1&0\\
0 & q
\end{smallmatrix}\right)\, R(x,y)=R(x,y)\, 
\left(\begin{smallmatrix}
1&0\\
0 & q
\end{smallmatrix}\right) \otimes
\left(\begin{smallmatrix}
1&0\\
0 & q
\end{smallmatrix}\right)\,.
\]
Proposition \ref{RMM} implies the following identity in $\End_{\bP[x_1,x_2,q]}(V\otimes V\otimes\mathbb{V}_N[x_1,x_2,q])$,
\begin{multline*}
R_{12}(x,y) 
\left(\begin{smallmatrix}
1&0\\
0 & q
\end{smallmatrix}\right)_1 M_1(x_1,t_1,\ldots,t_N)
\left(\begin{smallmatrix}
1&0\\
0 & q
\end{smallmatrix}
\right)_2 M_2(x_2,t_1,\ldots,t_N) =\\
\left(\begin{smallmatrix}
1&0\\
0 & q
\end{smallmatrix}\right)_1 
\left(\begin{smallmatrix}
1&0\\
0 & q
\end{smallmatrix}
\right)_2 R_{12}(x,y)M_1(x_1,t_1,\ldots,t_N)M_2(x_2,t_1,\ldots,t_N)=\\
\left(\begin{smallmatrix}
1&0\\
0 & q
\end{smallmatrix}\right)_1 
\left(\begin{smallmatrix}
1&0\\
0 & q
\end{smallmatrix}
\right)_2 M_2(x_2,t_1,\ldots,t_N) M_1(x_1,t_1,\ldots,t_N)R_{12}(x,y)=\\
\left(\begin{smallmatrix}
1&0\\
0 & q
\end{smallmatrix}
\right)_2 M_2(x_2,t_1,\ldots,t_N) 
\left(\begin{smallmatrix}
1&0\\
0 & q
\end{smallmatrix}\right)_1 
M_1(x_1,t_1,\ldots,t_N)R_{12}(x,y)\,.
\end{multline*}
Multiplying from the left with $R(x,y)^{-1}$ and taking the trace over $V\otimes V$ on both sides of the equality, the assertion follows for $*$. The arguments for $T'(x)=A\rq{}(x)+qD\rq{}(x)$  are completely analogous, or see \cite[Proposition 3.13]{GoKo}.
 \end{proof}

 \begin{remark}
 One can show by similar straightforward arguments that moreover the endomorphisms given by $T^{(i)}$ and $T'^{(j)}$ pairwise commute.
 \end{remark}
 
 \subsection{Equivariant quantum cohomology}
 We obtain a realisation of the  equivariant quantum cohomology $qH^*_T(\Xn)$.
\begin{theorem}
\label{QuantumH}
The algebra $H^*_T(\Xn)$ with the above deformed multiplication structures $*$ or $*'$ are both isomorphic to $qH^*_T(\Xn)$. In particular the quantum deformation of the multiplication by the equivariant Chern classes of the tautological (respective quotient) bundle is represented by the geometric convolution describing the operator $D(x)$ (and $D'(x)$ respectively).
\end{theorem}
\begin{proof} According to the known formulae for the equivariant quantum multiplication from \cite[Theorem 1]{Mich1}, Proposition~\ref{D} implies that the operator $D^{n-1}$ gives the quantum correction to the multiplication by the equivariant first Chern class of the tautological bundle, which can be easily connected to the multiplication by the Schubert divisor class, hence one can recover the equivariant quantum Pieri rule in our deformed ring. Now the main result of \cite{Mich1} says that the equivariant Pieri rule determines the quantum equivariant multiplication uniquely and the theorem follows.
\end{proof}

\begin{remark}
The existence of an isomorphism for $*'$ was already established purely combinatorially in  \cite{GoKo}, the geometric interpretation of the operators is new. It is interesting to compare this result with the "classical vs quantum" result form \cite{BM}.
\end{remark} 

\section{Connection with the current algebra $\mg[t]$}
\label{secYBAandcurrent}
In this section we finally connect our Yang-Baxter algebras $\op{YB}_N$ and $\op{YB_N^{\prime}}$ with the universal enveloping algebra $\cU$ of  $\mg[t]$ via some Schur-Weyl duality type result with  $\bH$. The general ideas behind our constructions are not new and used in a similar way already at several places in the geometric representation theory literature to construct quantum groups, in particular Yangians and quantum affine algebras. The new aspect here is however that we directly construct a Schur quotient of the {\it non-quantised}  universal enveloping algebra without having to invoke a limit argument or specialisation procedure to delete the deformation parameter. 

\subsection{$\mV_N$ as representation of $\bH$}
Recall the algebra $\bH$ from Definition~\ref{DefH}. Since $\bH$ is isomorphic to its opposite algebra via $w\mapsto w^{-1}$ for $w\in S_N$, we can (and will) use the previously defined left action of  $\bH$ on $\mV_N$ instead of a (for Schur-Weyl dualities more common) right action. 

\begin{definition}
For a composition $\ba=(a_1,\ldots, a_r)$ of $N$ let $S_{\ba}=S_{a_1}\times S_{a_2}\times\ldots\times S_{a_r}$ be the corresponding parabolic (or Young) subgroup of $S_N$ with idempotent
\begin{eqnarray*}
e_{\ba}&=&\frac{1}{|S_{\ba}|}\sum_{ w\in S_{\ba}}  w \quad \in S_\ba
\end{eqnarray*}
\end{definition}
In particular, if  $\ba=(N)$, then $e_\ba$ is the full symmetriser, whereas for $\ba=(1,\ldots ,1)$ we have $e_\ba=1$.  For any composition $\ba$ of $N$ we have $s_ie_\ba=e_\ba=e_\ba s_i$ if $s_i\in S_\ba$.
\begin{lemma}
\label{easylemma}
With the notation from Definition~\ref{DefH}, we have for any composition $\ba$ of $N$.
\begin{enumerate}[1.)]
\item The elements $wf=w\otimes f$, with $w\in S_N$ and $f$ from a fixed basis of $\bP$, form a basis of $\bH$. \label{basis} 
\item The invariant polynomials ${\bP}^{S_N}=\{f\in \bP\mid {}^w f=f\}$ form the centre  $\cZ=Z(\bH)$ of $\bH$. \label{center}
\item $\bH$ is free over $\cZ$ of rank $|S_N|^2$. \label{free}
\item $\bH e_{\ba}$ (respectively $e_{\ba}\bH$) has basis $fw$ (respectively $wf$), where $f$ through a basis of ${\bP}$ and $w$ through the minimal length coset representatives of $S_N/S_{\ba}$ (respectively of $S_{\ba}\backslash S_N$). \label{Heckemod}
\item $\bH e_{\ba}$ is free over $\mathcal{Z}$ of rank $|S_N||S_N/S_{\ba}|$.
\end{enumerate}
\end{lemma}
\begin{proof}
The first part is clear from Definition~\ref{DefH}. For the second one could apply  \cite[Theorem 6.5]{Lusztigcentre} or argue directly as follows: clearly $\bP^{S_N}\subseteq\cZ$ by definition of $\bH$. Conversely, let $z\in \cZ$, $z=\sum_{w\in S_N}f_ww$ for some $f_w\in\bP$. Then for any $i\in [N]$ we have $t_iz=\sum_{w\in S_N}t_if_w$ and $z_i=\sum_{w\in S_N}f_wwt_i=\sum_{w\in S_N}f_w\;{}^{w}\! t_iw$. 
Thus, by the first part  $f_w\;{}^{w}\!t_i=f_wt_i$ for all $i$ and $w$ such that $f_w\not=0$. This means however that $f_w=0$ for $w\not=1$. Thus $z\in\bP$, and then even $z\in\bP^{S_N}$ by definition of the multiplication in $\bH$.  
For third claim note that $\bH$ is free over ${\bP}$ of rank $|S_N|$ by definition, hence $\bH$ is free over ${\bP}^{S_N}$ of rank $|S_N|^2$ by invariant theory, \cite[18.3]{Kane}. The fourth part follows then directly from the first and the definition of $e_{\ba}$ using \cite[Proposition 1.10]{Humphreys} using the fact that for $w\in S_N$ the elements ${}^w\!f$ form a basis of $\bP$ if the $f$ do. Finally $\bH e_{\ba}$ is free over ${\bP}$ of rank $|S_N/S_{\ba}|$, hence is free over ${\bP}^{S_N}$ of rank $|S_N||S_N/S_{\ba}|$, again by \cite[18.3]{Kane}.
\end{proof}

\begin{proposition}
\label{Homs}
Let $\ba$, $\bb$ be compositions of $N$ and consider the space of $\bH$-module maps 
\begin{equation}\label{2}
\op{Hom}_{\bH}(\bH e_{\ba}, \bH e_{\bb})=e_{\ba} \bH e_{\bb}.
\end{equation}
\begin{enumerate}[1.)]
\item A $\mathbb{C}$-basis $\mathbb{B}$ of ($\ref{2}$) is given by all $e_{\ba}  w f e_{\bb}$ where $w$ runs through all minimal length double coset representatives in $S_{\ba}\backslash S_N/S_\bb$, and $f$ through a basis of ${\bP}^{S_\bb\cap  (w^{-1}S_{\ba}w)}$. \label{first}
\item In particular (\ref{2}) is free over $\cZ$ of finite rank.\label{twotwo}
\end{enumerate}
\end{proposition}
\begin{proof}
The identification $\eqref{2}$ is given by $\varphi\rightarrow \varphi(e_\ba)$, since $\varphi$ is determined by $\varphi(e_\ba)$ and $\varphi(e_\ba)=\varphi(e_\ba^2)=e_\ba\varphi(e_\ba)\in e_\ba \bH e_\bb$. The inverse map sends $x\in  e_\ba \bH e_\bb$ to the operator of right multiplication with $x$.
We first show that the elements in $\mathbb{B}$  span. Clearly, the $e_\ba  w f e_\bb$ for $w\in S_N$ and $f\in\bP$ span. Since $e_\ba  w f e_\bb = e_\ba s_i w f e_\bb=$ for any $s_i\in S_\ba$ and $e_\ba  w f e_\bb = e_\ba  w f s_ie_\bb=e_\ba w s_i {}^{s_i}\!f e_\bb$ for $s_i\in S_\bb$ we can assume $w$ is of the required form.  Let now $s_i\in S_\bb\cap  (w^{-1}S_\ba w)$ and $f\in {\bP}$. Write
$f=f_1+(x_i-x_{i+1})f_2$, with uniquely defined $s_i$-invariant polynomials $f_1,\,f_2\in {\bP}$. Then $x:=e_\ba\, w\,(x_i-x_{i+1})f_2\,e_\bb=0$,  since
\begin{equation*}
\begin{array}[t]{cclclcl}
x&=&e_\ba\, w \,(x_i-x_{i+1})f_2\,s_i\,e_\bb&=&-e_\ba w s_i\,(x_i-x_{i+1})f_2\,e_\bb\\
&=&-e_\ba w s_i w^{-1} w(x_i-x_{i+1})f_2\,e_\bb
&=&-e_\ba w (x_i-x_{i+1})f_2\,e_\bb&=&-x.
\end{array}
\end{equation*}
Hence $e_\ba w f\,e_\bb=e_\ba w f_1\,e_\bb$, thus we can assume $f$ to be $s_i$-invariant for any  $s_i\in S_\bb\cap  (w^{-1}S_\ba w)$ and so $\mathbb{B}$ spans. On the other hand $\mathbb{B}\subseteq e_\ba He_\bb\subseteq e_\ba H$ is a linearly independent subset thanks to Lemma ~\ref{easylemma} and again \cite[Proposition 1.10]{Humphreys}, and thus a basis. The second claim is now clear since $e_\ba He_\bb$ is free over ${\bP}^{S_\bb\cap  (w^{-1}S_\ba w)}$ of finite rank, hence free over $\cZ={\bP}^S_N$ of finite rank by invariant theory, \cite[18.3]{Kane}.
\end{proof}

\begin{lemma} 
\label{permmodule}
Let $0\leq k\leq N$, $\zeta=(0,\ldots,0,1,\ldots,1)\in \Lambda_n$ and $\ba=(k,N-k)=e_{(k,N-k)}$. With the action of $\bH$ from Proposition~\ref{Haction}, there is then an isomorphism of $\bH$-modules
\begin{eqnarray*}
(V[t])^{\otimes N}_k&\cong&\bH e_{\ba}, \quad\quad pv_\la\longmapsto pwe_{\ba},
\end{eqnarray*}
where $w\in S_N$ is such that $v_\la=v_{w^{-1}(\zeta)}$.
\end{lemma}
\begin{proof} The map is an isomorphism of vector spaces, even of $\bP$-modules, by Lemma~\ref{easylemma} part \ref{Heckemod}. Since $s_ipwe_{\ba}={}^{s_i}\!p\,s_iwe_{\ba}$, the $S_N$-action on $\bH e_\ba$ translates into the simultaneous permutation of the variables $t_1,\ldots,t_N$ and $\{0,1\}$ words labelling the basis vectors in $(V[t])^{\otimes N}_k$.
\end{proof}

\subsection{$\mV_N$ as representation of $\mg[t]$} 
We consider the complex Lie algebra $\mg$ with its standard basis $E=E_{1,2}$, $F=E_{2,1}$, $H_1=E_{1,1}$ and $H_2=E_{2,2}$ written in matrix units.   
\begin{definition}
Let $\mg[t]$ be the {\it current (Lie) algebra} for $\mg$, that is
$\mg[t]=\mg\otimes \Bbb C[t]$ as vector space, with Lie bracket defined for $x,y\in\mg$ $i,j\in \mathbb{Z}_{\geq0}$ as
$[x\otimes t^i,y\otimes t^j]=[x,y]\otimes t^{i+j}.$
\end{definition}
It acts on $V[t]$ in the obvious way, and on $V[t]^{\otimes N}$ by the usual comultiplication
$$\Delta(x\otimes t^a)=(x\otimes t^a)\otimes 1 + 1\otimes (x\otimes t^a),$$
where $1$ denotes the identity map. This action was considered explicitly e.g. in \cite{Varchenko}.

Let  $\cU=\cU(\mg[t])$ denote the universal enveloping algebra of $\mg[t]$. 
\begin{proposition} 
\label{surj}
The action map induces a surjective algebra homomorphism
\begin{eqnarray*}
\Psi:\;\cU&\longrightarrow& \End_{\bH} (V[t]^{\otimes N}).
\end{eqnarray*}
\end{proposition}

Note that the $\cU$-action obviously commutes with the $\bH$-action, hence the map is well-defined. We prepare the rest of the proof with the following easy fact:
\begin{lemma}
\label{isovect}
There is an isomorphism of vector spaces
\begin{eqnarray}
\label{toshow}
\End_{\Bbb C[t]}(V[t])^{\otimes N}&\cong &\End_{\bP}(V^{\otimes N}\otimes {\bP})\\
{\bf A}:=A_1\otimes\ldots\otimes A_d&\mapsto& f_{\bf A}\nonumber
\end{eqnarray}
with $f_{\bf A}(v_{i_1}\otimes\ldots\otimes v_{i_d}) = A_1v_{i_1}\otimes\ldots\otimes A_{i_d}v_d$, and the identification $V^{\otimes N}\otimes {\bP}=
(V[t]\otimes\ldots\otimes V[t])\otimes_{\bP}{\bP}.$
\end{lemma}
\begin{proof}
A $\Bbb{C}$-basis of $\End_{\Bbb C[t]}(V[t])$ is given by $E_{a,b}t^r$ where $1\leq a,b\leq 2$, $r\geq 0$, and their $N$-fold tensor products $B_{\bf a,\bf b}=E_{a_1b_1}t^{r_1}\otimes\ldots\otimes E_{a_N,b_N}t^{r_N}$ form a basis of the left hand side in \eqref{toshow}. The right hand side has basis
$E_{\bf a,\bf b}p$ where ${\bf a},{\bf b}\in \Lambda$ and $p$ runs through a basis of $\bP$. Then $B_{\bf a,\bf b}\mapsto E_{\bf a,\bf b}p$, where $p=t_1^{r_1}\cdots t_N^{r_N}$ defines an isomorphism of $\Bbb{C}$-vector spaces as claimed. 
\end{proof}
\begin{proof}[Proof of Proposition~\ref{surj}]

Note first that 
$\End_{\bH} (V[t]^{\otimes N})=\End_{\bH}(V^{\otimes N}\otimes {\bP})$ can be identified via Lemma~\ref{isovect} with $(\End_{\bP} (V^{\otimes N}\otimes {\bP}))^{S_N}$ where $S_N$ acts on endomorphisms $f$ as $(w\cdot f)(x)= w f(w^{-1}(x))$ for $w\in S_N$. By Lemma~\ref{isovect} this is isomorphic to $(\End_{\Bbb C[t]}(V[t])^{\otimes N})^{S_N}$ where the $S_N$-action is just the permutation of the tensor factors, such that $s_i$ swaps the $i$th and $(i+1)$th factor. 
Hence to establish the claim, is suffices to show that the $S_N$-invariants in $\End_{\Bbb C[t]}(V[t])^{\otimes N}$ are in the image of $\Psi$. By polarisation, \cite[Lemma B.2.3]{GW}, it is even enough to see that all $A\otimes\ldots\otimes A$ with $A\in\mg$ are in the image of $\Psi$. By definition elements of the form $\Delta(xt^a)$, $x\in \mg$, $a\in \Bbb Z_{\leq 0}$ are in the image. 
One easily verifies inductively that any $A\otimes\ldots\otimes A$ is contained in the subalgebra generated by these elements. For $N=1$ there is nothing to check, and then for instance 
\begin{eqnarray*}
A\otimes A&=&\frac{1}{2}\left(\Delta (A)^2-\Delta (A^2)\right),\\
A\otimes A\otimes A&=& \frac{1}{6}\left(\Delta (A)^3 - 3\Delta (A^2)\Delta (A) + 2\Delta (A^3)\right).
\end{eqnarray*}
Hence, $\Psi$ is a surjective algebra homomorphism. 
\end{proof}
\begin{remark}
We consider in Proposition~\ref{surj} only the case of $\mg[t]$ relevant to our setup, but the analogous statement holds for $\mathfrak{gl}_n[t]$ for any $n\geq 2$ with the obvious generalisation of the arguments. The vector space $V$ should be replaced by the natural representation of $\mathfrak{gl}_n$ and one should work with the modules $\bigoplus_\lambda e_{\ba}H$ where $\ba$ runs through all compositions of $N$ with not more than $n$ parts.
\end{remark}
\subsection{Localised Schur algebra and Yang-Baxter algebras}
\begin{definition}
Let $\cS\subset\cU$ be the subalgebra generated by all $(H_1+H_2)\otimes t^r\in\mg[t]$, $r\geq 0$. 
By definition $\cS$ is multiplicatively closed and central in $\cU$. Let $\cU_{\cS}$ be the (Ore) localisation of $\cU$ at $\cS$. (That means we formally make the elements in $\cS$ invertible).
\end{definition}
Note that $(H_1+H_2)\otimes t^r$ acts by multiplication with the $r$-th symmetric power sum $p_r(t_1,t_2,\ldots, t_N)=t_1^r+\ldots+t^r_N$ on $V^{\otimes N}\otimes {\bP}$. Since the $p_r$ generate the algebra of symmetric functions,  the image of $\cS$ under $\Psi$ are all endomorphisms given by multiplication with a symmetric polynomial, in particular $\Psi(\cS)$ agrees with the image of the action of $\cZ$ via the identification from Lemma~\ref{easylemma}, part \ref{center}. Let $\bH_\cZ$ denote the localisation of $\bH$ at $\cZ$ and consider the $\bH_\cZ$-module 
\begin{eqnarray}
\label{Vlocalisation}
(V^{\otimes N} \otimes\bP)_{\cZ}=(V^{\otimes N} \otimes\bP)_{\bP^{S_N}}
 \end{eqnarray}
   obtained by localisation at $\cZ=\bP^{S_N}$.  Since $(V^{\otimes N} \otimes\bP)$ is free over $\cZ$ by Lemma~\ref{permmodule} and Lemma~\ref{easylemma} of finite rank, localisation at $\cZ$ behaves well and we have canonical isomorphisms of algebras
$$\End_{\bH_{\cZ}}( (V^{\otimes N} \otimes\bP)_{\cZ})\cong\End_{\bH_{\cZ}}( (V^{\otimes N} \otimes\bP))_\cZ\cong\End_{\bH_{\cZ}}( V^{\otimes N} \otimes(\bP_{\cZ})). 
$$

Proposition~\ref{surj} implies now directly the following statement.
\begin{corollary}
\label{UmodI}
The algebra homomorphism $\Psi$ induces a surjective algebra homomorphism, $\mathcal \cU_{\cS}
\longrightarrow
\End_{\bH_{\cZ}}( V^{\otimes N} \otimes(\bP_{\cZ}))$. In particular
\begin{eqnarray*}
\End_{\bH_{\cZ}} (V^{\otimes N} \otimes(\bP_{\cZ}))&\cong& \mathcal \cU_{\cS}/I
\end{eqnarray*}
for some two-sided ideal $I$ in  $\cU_{\cS}$.
\end{corollary}
Because of Lemma~\ref{permmodule} we call $\mathcal \cU_{\cS}/I$ the {\it localised Schur algebra} for $\mg$ in analogy to the classical Schur algebra for $\mathfrak{gl}_n$, see \cite[Section 4]{Mathas}.
Passing to the localisation provides a natural framework with the fixed point or Bethe basis available, but at the same time a natural appearance of the denominators in terms of $\bH$.

We finally connect the Yang-Baxter algebras with the universal enveloping algebra $\cU$ of  $\mg[t]$. 

\begin{definition} Let $\mathcal{Y}_N$ be the subalgebra of endomorphisms of $\mV_N$ generated by both Yang-Baxter algebras, $\op{YB}_N$ and $\op{YB}'_N$, localised at the subalgebra generated by the $A(x)$ and $A'(x)$. 
\end{definition}

\begin{theorem}
\label{YBAandcurrent}
There is an isomorphism of algebras 
\begin{eqnarray}
\label{iso}
\mathcal{Y}_N&\cong& \mathcal  \mathcal \cU_{\cS}/I 
\end{eqnarray}
\end{theorem}
\begin{proof}
Since the Yang-Baxter algebras commute with the $\bH$-action by Proposition~\ref{Propcommute}, they induce well-defined endomorphisms in $\End_{\bP}(V^{\otimes N}\otimes {\bP})$ and by Corollary~\ref{explicitformulas} it extends to a well-defined action of $\mathcal{Y}_N$ on $(V^{\otimes N} \otimes\bP)_{\cZ}$. Hence by Corollary~\ref{UmodI} there is an embedding of algebras
\begin{eqnarray*}
\mathcal{Y}_N&\hookrightarrow&\mathcal\mathcal \cU_{\cS}/I. 
\end{eqnarray*}
 To see that this is an isomorphism it is enough to show that the image contains the generators $E\otimes t^j$, $F\otimes t^j$, $H_1\otimes t^j$ and $H_2\otimes t^j$ for any $j\geq 0$ and the inverses of $\bP^{S_N}$. Apart from the $E\otimes t^j$ and  $F\otimes t^j$ this follows from the explicit formulae in Section~\ref{sec:explicitformulas} for the $A$-operators.  For the remaining operators we use the geometric Definition~\ref{ABCDgeom} of the operators. Since the classes of the tautological and quotient bundles generate the complete cohomology ring, we can in particular obtain $b_n$ and $c_n$via  $\op{YB}_N$, and hence $E\otimes 1$ and  $F\otimes 1$. Using the $A$-operators again we obtain also all other $t$-powers, see \cite[Appendix]{Varchenko} for explicit formulae. 
\end{proof}

\begin{remark}
Observe, that we need both Yang-Baxter algebras to obtain the Schur-algebra. It would be interesting to see if there is a general theory behind this.
\end{remark}

\section{Connection to COHAs}
In this section we will briefly explain the connection of our Yang-Baxter algebra with another important algebra, the cohomological Hall algebra (short COHA) introduced in \cite{KonS}, \cite{S}. The COHA is an associative algebra attached to a quiver (with possibly a potential).  The underlying space is the direct sum of the cohomology of the quotient stacks of isomorphism classes of representations for each fixed dimension vector. 

\subsection{The special example of interest: $\CH$}
For our setup it is enough to consider the easiest COHA, we call it $\CH$, of type $A_1$.  That is an algebra  attached to the one point and no loops quiver $Q$.  This quiver occurs because we work with Grassmannians only, instead of more general flag varieties. (The Dynkin type of $Q$ is the type of the current Lie algebra in Section~\ref{secYBAandcurrent}).

We consider for a fixed natural number $d$ the space $\op{M}_{d}$ of representations of $Q$, which is just the vector space of complex $d\times d$-matrices with the action of $\op{G}(d):=\op{GL}_d(\mathbb{C})$ by conjugation. The quotient stack is then given by conjugacy classes of matrices.

Then for $d_1,d_2\in \mathbb{Z}_{\geq 0}$ and $d=d_1+d_2$ consider the subspace $\op{M}_{d_1,d_2}\subseteq  \op{M}_{d}$ of upper triangular block matrices with block size $d_1,d_2$ together with the maps $p_i$ $i=1,2$, picking out the blocks  
\begin{eqnarray}
\label{matrices}
  \op{M}_{d_i}\stackrel{p_i}{\longleftarrow}\quad \op{M}_{d_1,d_2}\quad\stackrel{\op{incl}}\longrightarrow\op{M}_{d}&&D_i\leftmapsto\begin{pmatrix}D_1&B\\0&D_2\end{pmatrix}\;\;\in\;\;  \op{M}_{d_1+d_2}
\end{eqnarray}

Let $\op{BG}(d)$ be the classifying space of  $\op{G}(d)$. 
We choose the standard  model $\op{BG}(d)=\varinjlim_N \op{Gr}(d,\mC^N)=\op{Gr}(d,\mC^\infty)$, with the obvious embedding $\mathbb{C}^N\subseteq\mathbb{C}^{N+1}$ and identify $ H^*_{\op{G}(d)}(\op{M}_{d})= H^*_{\op{G}(d)}(\op{pt})= H^*(\op{BG}(d)).$

\begin{definition}
The COHA $\CH$ is the algebra with underlying (graded) vector space 
\begin{eqnarray}
  \CH&=&\bigoplus_{d\geq 1} H^*_{\op{G}(d)}(\op{M}_{d})\;=\;\bigoplus_{d\geq 1} H^*(\op{BG}(d)), 
  \end{eqnarray}
and the (graded) multiplication $H^*_{\op{G}(d_1)}(\op{M}_{d_1})\otimes H^*_{\op{G}(d_1)}(\op{M}_{d_2})\rightarrow H^*_{\op{G}(d_1+d_2)}(\op{M}_{d_1+d_2})$ defined as the push-forward in cohomology of the corresponding quotient stacks arising from the diagram \eqref{matrices}.  Explicitly this is given by the map $\op{m}_*$, where
\begin{eqnarray}
\op{m}:\quad \op{BG}(d_1)\times \op{BG}(d_2)\rightarrow \op{BG}(d_1+d_2), 
\end{eqnarray}
is the canonical map arising from the embedding $\op{G}(d_1)\times \op{G}(d_2)\subseteq \op{G}(d)$.
\end{definition}
We refer to see \cite{Xiao2} for explicit formulae. In contrast to the general case, $\CH$ has an explicit presentation, see \cite[Section 2.5]{KonS}. Namely,  $\CH$ is isomorphic to the infinite exterior algebra with generators $\psi_{2j+1}$ $j\in \Bbb \mathbb{Z}_{\geq0}$ corresponding to the basis vectors $x^j\in  H^*(\op{BG}(1))=H^*(\mathbb{C}P^\infty)\cong\mathbb{C}[x]$. Monomials of degree $d$ in these generators correspond to Schur polynomials in $H^*_{\op{G}(d)}(\op{pt})$, and suggest a connection with Schubert calculus.  

\subsection {The action on $H^*_T$}

For any partial flag variety $\op{X}_{(n,r,N-n-r)}$ we have be the classifying map 
$\xi: \op{X}_{(n,r,N-n-r)}\rightarrow \op{BG}(r)$ (explicitly it sends a partial flag 
$(F_n\subset F_{n+r}\subset \mathbb{C}^N)$ to $F_{n}^\perp\subset \mathbb{C}^N$ in $\op{BG}(r)$, where we take the orthogonal complement $F_{n}^\perp$ in  $F_{n+r}$ with respect to the standard scalar product).

\begin{proposition} 
\label{COHAactions}
There are two actions of $\CH$ on $H^*_T=\bigoplus_{n=0}^N H_T^*(\Xn)$:
The assignment $\psi_{2j+1}\longmapsto \gamma^+_j$ respectively the assignment $\psi_{2j+1}\longmapsto \gamma^-_j$ with $j\in \mathbb{Z}_{\geq0}$, where 
\begin{eqnarray}
\label{CPaction}
\gamma^+_j={\pi_2}_{*}( (\xi^* (x^j)\cdot) {\pi_1^*})&\text{and}&\gamma^-_j={\pi_1}_{*}((\xi^* (x^j)\cdot) {\pi_2^*}).
\end{eqnarray}
with the notation as in \eqref{correspondence} defines an action on $\bigoplus_{n=0}^N H^*(\Xn)$ (acting on all summands) which extends to $H^*_T$ by taking the equivariant version of $\xi^* (x^j)$ instead.
\end{proposition}

Before we state and prove the proposition, we like to stress that our Grassmannians or the related partial flag varieties carry a natural $T$-action which allows us to consider their $T$-equivariant cohomology. There is no natural action of this torus on the spaces $\op{BG}(d)$ which go into the definition of the COHA. However, $\CH$ will act via twisting the convolutions from \eqref{correspondence} by the Chern classes of $G$-equivariant, and therefore $T$-equivariant, bundles. Thus, the above formulae give well-defined operators on $H^*_T$. In the non-equivariant setting the result can already be found in \cite{Xiao1}, \cite[Proposition 4.1.1]{S} and also deduced as a special case from \cite{Franzen}. We provide the full arguments which then extend to the equivariant case. Consider the following commutative diagram, with the obvious (and all proper) maps (and $r$, $s$ natural numbers such that the spaces make sense).
\begin{equation}
\label{big}
\begin{tikzcd}
{}&{}&\op{X}_{(n,r,s,N-n-r-s)} \arrow[bend left,red]{ddd}{p} \arrow{dl}{\pi_5} \arrow{dr}{\pi_6} &\\
{} & \op{X}_{(n,r,N-n-r)} \arrow{dr}{\pi_2} \arrow{dl}{\pi_1}&& \op{X}_{(n+r,s,N-n-r-s)}\arrow{dr}{\pi_4}\arrow{dl}{\pi_3}\\
\op{X}_n && \op{X}_{n+r} && \op{X}_{n+r+s}&\\
{} & {} & \op{X}_{(n,r+s,N-n-r-s)}  \arrow{urr}{\pi_8} \arrow{ull}{\pi_7} &
\end{tikzcd}
\end{equation}
\begin{proof}
We first prove the non-equivariant version. We only deal with the first case, the second works completely analogously.  We claim that ${\pi_2}_{*}( (\xi^* (x)\cdot) {\pi_1^*})$ with $x\in\op{BG}(r)$ defines the first action. Clearly, it satisfies \eqref{CPaction}. To check that it is indeed an action, consider the following classifying maps (for arbitrary $s$ such that the spaces make sense)
\begin{equation}
\begin{tikzcd}
 \op{X}_{(n,r,N-n-1)}\arrow{d}{\xi_1}&& \op{X}_{(n+1,s,N-n-2)}\arrow{d}{\xi_2}&& \op{X}_{(n,r+s,N-n-2)}\arrow{d}{\xi_3}\\
\op{BG}(r)&\times& \op{BG}(s)&\stackrel{\op{m}}\longrightarrow&\op{BG}(r+s)
\end{tikzcd}
\end{equation}
Then it is enough to verify the following claim:
Let $a=\xi_1^*(x_r)$, $b=\xi_2^*(x_s)$ for some $x_r\in \op{BG}(r)$, $x_s\in \op{BG}(s)$, and set $c=\xi_3^*(\op{m}_*(x_r\otimes x_s))$. Then it holds  
${\pi_{4}}_*(b\cdot\pi_3^*{\pi_2}_*(a\cdot \pi_{1}^*(y)))={\pi_{8}}_*(c\cdot\pi_{7}^*(y))$ for all $y\in H^*(\Xn)$.
To verify this we calculate
\begin{eqnarray*}
&{\pi_{4}}_*(b\cdot\pi_3^*{\pi_2}_*(a\cdot \pi_{1}^*(y)))
={\pi_{4}}_*(b\cdot{\pi_6}_*({\pi_5}^*(a)\cdot {\pi_5}^*(\pi_{1}^*(y))))
={\pi_{4}}_*{\pi_6}_*({\pi_5}^*(a)\cdot \pi_6^*(b)\cdot {\pi_5}^*\pi_{1}^*(y))&\\
&={\pi_{8}}_*{p}_*({\pi_5}^*(a)\cdot \pi_6^*(b)\cdot {p}^*\pi_{7}^*(y))
={\pi_{8}}_*({p}_*({\pi_5}^*(a)\cdot \pi_6^*(b))\cdot\pi_{7}^*(y))
={\pi_{8}}_*(c\cdot\pi_{7}^*(y)),&
\end{eqnarray*}
where we used first proper base change and the fact that $\pi_5^*$ is a ring homomorphism, then the projection formula, the commutativity of \eqref{big}, again the projection formula, and finally the definition of $c$ together with $\xi_3^*\op{m}_*=p_* (\xi_1\pi_5 \otimes \xi_2\pi_6)^*:H^*(\op{BG}(r))\otimes H^*(\op{BG}(s))\rightarrow H^*(\op{X}_{(n,r+s,N-n-2)}).$ The claim follows. In particular the action respects the exterior algebra relations in $\CH$. This provides a well-defines action in the non-equivariant case. Let now $t\in T$ and $a_t$ the action map by $t$, then by definition $\xi a_t=\xi$ and so $a_t^*(\xi^*(x^j)\cdot y)=a_t^*(\xi^*(x^j))\cdot a_t^*( y)=(\xi^*(x^j))\cdot a_t^*( y)$ for $y\in H^*_T$, and then also similarly for products in the generators $x^j$. Hence the equivariant version is well-defined as well and defines again an action of $\CH$. 
\end{proof}

\subsection{COHA action and geometric Yang-Baxter algebra action}
The connection to the Yang-Baxter algebras is given as follows.
\begin{theorem}
\label{ThmgammasYB}
The operators $\gamma^+_j$ and $\gamma^-_j$ from Proposition~\ref{COHAactions}  belong to the Yang-Baxter algebra  $\mathcal{YB}_N$ respectively  $\mathcal{YB}_N'$.
\end{theorem}
\begin{proof}
We use again the notation from \eqref{correspondence}. Then there is a short exact sequence of vector bundles on $\op{X}_{(n,1,N-n-1)}$ of the form
\begin{equation*}
0\rightarrow \pi_1^*(\mathcal{T}_n)\rightarrow \pi_2^*(\mathcal{T}_{n+1})\rightarrow M_2\rightarrow 0,
\end{equation*}
where $M_2$ is the tautological line bundle. Let $e=c_1(M_2)$ be its first Chern class.
Then by the Whitney sum formula and the definitions, $e=c_1(\pi_2^*(\mathcal T_{n+1}))-c_1(\pi_1^*(\mathcal{T}_n))=\pi_2^*(c_1(\mathcal T_{n+1}))-\pi_1^*(c_1(\mathcal T_n))$. Therefore, we have for any $x\in H^*(\Xn)$ (by the projection formula and functoriality)
\begin{eqnarray*}
{\pi_2}_*(e\cdot \pi_1^*(x))&=&{\pi_2}_*(\pi_2^*(c_1(\mathcal T_{n+1})\cdot \pi_1^*(x))-{\pi_2}_*(\pi_1^*(c_1(\mathcal T_n))\cdot \pi_1^*(x))\\
&=&c_1(\mathcal T_{n+1})\cdot{\pi_2}_*( \pi_1^*(x))-{\pi_2}_*(\pi_1^*(c_1(\mathcal T_n)\cdot x))=
c_1(\mathcal T_{n+1})\cdot b_n(x)-b_n(c_1(\mathcal T_n)\cdot x)
\end{eqnarray*}
and the latter is by definition in $\mathcal{YB}_N$. Hence  $\gamma^{+}_1$ is contained in  $\mathcal{Y}_N$, similarly for any $\gamma^{+}_j$ by taking powers of $e$. The argument for the $\gamma^{-}_j$ is analogous.
\end{proof}
\begin{remark}
Likewise ${\pi_2}_*(P\cdot \pi_1^*)$ is contained in $\mathcal{YB}_N$ for any polynomial $P$ in the $c_i(\mathcal{T}_n)$'s and $e$. In particular, the action map of $a_{21}\otimes t^k\in \mathfrak{gl}_2[t]$ from the representation of $\mathfrak{gl}_2[t]$ described in \cite{Varchenko}, \cite{Vasserot} belongs to  $\mathcal{YB}_N$. Similarly, $a_{12}\otimes t^k\in \mathfrak{gl}_2[t]$ acts by an operator in $\mathcal{YB}_N$.
\end{remark}

\begin{remark}
Calculating explicitly the maps $\gamma^{\pm}_j: H^*_T(\Xn) \rightarrow H^*_T(\op{X}_{n\pm1})$  is easy in the fixed point basis using the definitions and \eqref{AB}. Via the identification from Corollary~\ref{Bethe} and \eqref{beta}, we obtain for the vector $\bb_\zeta$ with $\zeta=(0,\ldots,0, 1,\ldots,1)$ the following formulae,
\begin{eqnarray}
\label{coha}
\gamma_j^+ (\bb_ \zeta)\;=\;\sum_{i=1}^k\prod_{r=k+1}^{N}
\frac{t^j_i}{t_{r}-t_{i}}~\bb_{0\ldots 0\underset{i}{1}0\ldots 01\ldots  1}&\text{and}&
\gamma^-_j( \bb_\zeta) \;=\;\sum_{i=k+1}^{N}\prod_{r=1}^{k}\frac{
t^j_i}{t_{i}-t_{r}}
~\bb_{0\ldots 01\ldots 1\underset{i}{0}1\ldots 1}.
\end{eqnarray}
Thanks to Theorem~\ref{ThmgammasYB} and Proposition~\ref{Propcommute} the maps $\gamma_j^\pm$ are $S_N$-equivariant and hence \eqref{coha} provides explicit formulae for all basis vectors. 
\end{remark}

 \input{references}

\end{document}

%% file: references.tex
\addtocontents{toc}{\protect\setcounter{tocdepth}{0}}
\bibliographystyle{alpha}